\documentclass{article}
\usepackage{hchang}
\nolinenumbers 
\usepackage{hyperref}
\usepackage{comment}
\usepackage{tcolorbox}
\usepackage{color}
\usepackage{graphicx} 
  \setkeys{Gin}{width=\linewidth,totalheight=\textheight,keepaspectratio}
  \graphicspath{{graphics/}} 
\usepackage{booktabs} 
\usepackage{units}    
\usepackage{multicol} 
\usepackage{lipsum}   
\usepackage{fancyvrb} 
  \fvset{fontsize=\normalsize}
\usepackage{cleveref}
\usepackage{rotating}
\usepackage{transparent}
\usepackage{multicol}
\usepackage{tikz}
\usepackage{fullpage}
\usepackage{soul}
\usepackage{thmtools}
\usepackage{thm-restate}
\usetikzlibrary{arrows.meta,calc,decorations.pathreplacing,fit,positioning}

\makeatletter
\newenvironment{namedproof}[1][\proofname]{\par
  \normalfont \topsep6\p@\@plus6\p@\relax
  \trivlist
  \item[\hskip\labelsep
        \bfseries
    #1\@addpunct{:}]\ignorespaces
}
\makeatother

\newtheorem{lemma}{Lemma}
\newtheorem{theorem}{Theorem}
\newtheorem{claim}{Claim}
\newtheorem{definition}{Definition}

\newtheorem{question}{Question}

\newcommand{\eps}{\epsilon}

\newcommand{\pr}{\mathrm{Pr}}

\newcommand{\bfs}{\mathrm{BFS}}
\newcommand{\kpr}{\mathrm{KPR}}

\newcommand{\angbracket}[1]{\langle #1\rangle}
\newcommand{\Cc}{\mathcal C}

\newcommand{\Ss}{\mathcal S}
\newcommand{\Tt}{\mathcal T}

\renewcommand{\leq}{\leqslant}
\renewcommand{\le}{\leqslant}
\renewcommand{\geq}{\geqslant}
\renewcommand{\ge}{\geqslant}

\Crefname{figure}{Algorithm}{Algorithms}

\title{Separator Theorem for Minor-Free Graphs in Linear Time\thanks{A conference version of this paper will appear in the proceedings of the 58th Annual ACM Symposium on Theory of Computing (STOC 2026)~\cite{ourSTOC}.}}
\author{Édouard Bonnet\thanks{CNRS, ENS de Lyon, Université Claude Bernard Lyon 1, LIP UMR 5668, 69342 Lyon. \texttt{edouard.bonnet@ens-lyon.fr}}
\and
Tuukka Korhonen\thanks{University of Copenhagen. \texttt{tuko@di.ku.dk}}
\and
Hung Le\thanks{University of Massachusetts, Amherst. \texttt{hungle@cs.umass.edu}}
\and
Jason Li\thanks{Carnegie Mellon University. \texttt{jmli@cs.cmu.edu}}
\and
Tomáš Masařík\thanks{University of Warsaw. \texttt{masarik@mimuw.edu.pl}}
}
\date{}

\begin{document}

\maketitle
\begin{abstract}
    The planar separator theorem by Lipton and Tarjan [FOCS '77, SIAM Journal on Applied Mathematics '79] states that any planar graph with $n$ vertices has a balanced separator of size $O(\sqrt{n})$ that can be found in linear time. This landmark result kicked off decades of research on designing linear or nearly linear-time algorithms on planar graphs. In an attempt to generalize Lipton--Tarjan's theorem to nonplanar graphs, Alon, Seymour, and Thomas [STOC '90, Journal of the AMS '90] showed that any minor-free graph admits a balanced separator of size $O(\sqrt{n})$ that can be found in $O(n^{3/2})$ time. The superlinear running time in their separator theorem is a~key bottleneck for generalizing algorithmic results from planar to minor-free graphs. Despite extensive research for more than two decades, finding a balanced separator of size $O(\sqrt{n})$ in (linear) $O(n)$ time for minor-free graphs has remained a major open problem. Known algorithms either give a separator of size much larger than $O(\sqrt{n})$ or have superlinear running time, or both.
    
    In this paper, we answer the open problem affirmatively. Our algorithm is very simple: it runs a vertex-weighted variant of breadth-first search (BFS) a constant number of times on the input graph. Our key technical contribution is a weighting scheme on the vertices to guide the search for a balanced separator, offering a new connection between the size of a balanced separator and the existence of a clique-minor model. We believe that our weighting scheme may be of independent interest.
\end{abstract}

\section{Introduction}

In the late 70s, Lipton and Tarjan~\cite{LT79} introduced a \emph{planar} separator theorem in \emph{linear time}: there is a linear-time algorithm that, given any $n$-vertex planar graph, returns a~balanced separator of size $O(\sqrt{n})$. (Their separator bound improved upon the earlier bound of $O(\sqrt{n}\log^{3/2} n)$ by Ungar~\cite{Ungar1951}.) In a follow-up paper~\cite{LT80}, they gave a plethora of algorithmic applications of their separator theorem, from approximating NP-hard problems to data structures, circuit lower bounds, and maximum matching problems, to name a few. Their results have unleashed decades of intensive research on planar graph algorithms, continuing to this day. (See the book draft by Klein and Mozes~\cite{KM24} for a sample of results on planar graphs.) There is hardly any result on planar graphs that does not use the separator theorem or its variants~\cite{Miller86,Thorup04}. Clearly, the linear running time is crucial for many algorithmic applications.

The planar separator theorem by Lipton and Tarjan has motivated a long line of research on separator theorems for nonplanar graphs. Note that constant-degree expander graphs only have balanced separators of size $\Omega(n)$, and therefore, one has to impose additional structures besides sparsity on the input graph to guarantee the existence of a separator of sublinear size. A well-studied and perhaps most natural approach is to study graphs excluding a~fixed minor. We say that a graph $H$ is a \EMPH{minor} of $G$ if $H$ can be obtained from $G$ by a sequence of vertex and edge deletions, and edge contractions. We say that $G$ is \EMPH{$H$-minor-free} if it does not contain $H$ as a minor. Wagner's theorem~\cite{Wagner37} characterizes planar graphs in terms of forbidden minors: a graph is planar if and only if it excludes $K_5$ and $K_{3,3}$ as minors. Therefore, planar graphs form a subclass of $K_5$-minor-free graphs. This point of view naturally raises the following questions:

\begin{question}\label{question:main}
Can we construct a balanced separator of size $O(\sqrt{n})$ for any given $K_h$-minor-free graph in $O(n)$ time for any fixed $h$? Can the dependence on $h$ in the separator size and the running time be made polynomial?
\end{question}

A~resolution of the first part in \Cref{question:main} casts the separator theorem for planar graphs by Lipton and Tarjan as a~special case (when $h=5$) of a~much broader landscape. The second part seeks a stronger guarantee for a~more practical purpose: by imposing a polynomial dependence on $h$, one has to avoid sophisticated approaches based on the Robertson--Seymour structure theorem~\cite{RS03}, such as in~\cite{KR10}, which are often considered highly impractical~\cite{LR10}.

Partial progress on this question was made on graphs embeddable on a surface of genus $g$; such graphs exclude $K_{O(\sqrt{g})}$ as a minor. Indeed, the paper by Lipton and Tarjan~\cite{LT80} constructed a separator of size $O(g\sqrt{n})$. Gilbert, Hutchinson, and Tarjan~\cite{GHT84} improved the separator size to $O(\sqrt{gn})$ and gave an algorithm that runs in $O(gn)$ time, \emph{given a genus-$g$ embedding} of the input graph. However, finding such an embedding in linear time is difficult: when $g$ is part of the input, it is NP-hard~\cite{Thomassen89}. For a small genus $g$, an algorithm with running time $2^{O(\poly(g))}n$ was found decades later~\cite{mohar1996,KMR08}. To find an~$O(\poly(g)\sqrt{n})$-size separator, it suffices to embed a~genus-$g$ graph $G$ into a~surface with genus $\poly(g)$. There is a polynomial-time algorithm~\cite{KS15} (even when $g$ is unbounded) with such a guarantee, but it is unclear if the algorithm can be implemented in $O(\poly(g)n)$ time. In other words, even in the special case of genus-$g$ graphs, \Cref{question:main} still does not have a satisfactory answer. Note that a $K_5$-minor-free graph can have genus $g=\Omega(n)$, and therefore the class of $K_h$-minor-free graphs for a fixed $h$ is vastly broader than bounded-genus graphs.

In a breakthrough paper, Alon, Seymour, and Thomas~\cite{AST90Conf,AST90} made significant progress toward \Cref{question:main} by showing that any $K_h$-minor-free graph has a balanced separator of size $O(\poly(h)\sqrt{n})$. Their existential proof can then be turned into a polynomial-time algorithm.

\begin{theorem}[Alon, Seymour, and Thomas~\cite{AST90Conf,AST90}]\label{thm:AST}
For any integer $h\geq 1$, any $K_h$-minor-free graph admits a balanced separator of size $O(h^{3/2}\sqrt{n})$ that can be found in time $O(h^{1/2}\sqrt{n} m)$.
\end{theorem}

In concluding their paper, Alon, Seymour, and Thomas~\cite{AST90} wrote: \emph{``It may be that by using more sophisticated, dynamic data structures, the algorithm can be implemented more efficiently, but at the moment we do not see how to do so.''} Indeed, subsequent efforts to obtain linear (or almost linear) time algorithms are by designing completely different algorithms. Reed and Wood~\cite{RW09} gave an algorithm that finds a balanced separator of size $2^{O(h^2)}n^{2/3}$ in time $2^{O(h^2)}n$; these bounds are part of a more general trade-off between separator size and running time. Ignoring the dependence on $h$, their separator size is $O(n^{2/3})$, which is much larger than the $O(\sqrt{n})$ bound in \Cref{thm:AST}. Kawarabayashi and Reed~\cite{KR10} used Robertson--Seymour's structure theorem for minor-free graphs, giving a~quadratic algorithm for finding a separator of size $O(h\sqrt{n})$, and sketching another algorithm in time $O(g(h)n^{1+\eps})$ for any fixed constant $\eps \in (0,1)$. One of the authors~\cite{Kawara11} later clarified that the precise separator bound is $O(h\sqrt{n} + f(h))$. Here $f(h)$ and $g(h)$ are functions of the Robertson--Seymour type: a tower of several exponentials. Their algorithm is very complicated, along the lines of Robertson--Seymour decomposition. The full version has not yet appeared after over fifteen years.

Plotkin, Rao, and Smith~\cite{PRS94} introduced the notion of \emph{shallow minor}\footnote{Interestingly, it turns out that shallow minors play a crucial role in the recent theory of sparse graphs~\cite{NO12}.} and derived a~polynomial-time algorithm for finding a balanced separator for $K_h$-minor-free graphs of size
\begin{equation}\label{eq:shallow}
   O(n/\ell + h^2\ell \log n)
\end{equation}
where $\ell \in [1,n]$ is a parameter. For example, by setting $\ell = \sqrt{n}/(h\sqrt{\log n})$, we obtain a separator of size $O(h\sqrt{n\log n})$. While the running time of their algorithm~\cite{PRS94} is not better than \Cref{thm:AST}, it is amenable to a fast implementation. Specifically, Wulff-Nilsen~\cite{WulffNilsen11} designed a faster implementation of the algorithm by Plotkin, Rao, and Smith, obtaining two interesting trade-offs: a balanced separator of size $O(h\sqrt{n\log n})$ in time $O(\poly(h)n^{5/4+\eps})$ or of size $O(\poly(h)n^{4/5+\eps})$ in linear time $O(\poly(h)n)$. In a follow-up paper, Wulff-Nilsen~\cite{wulffnilsen14} improved the running time further when $\ell \approx n^{\eps}$; however, the improvement is not significant when $\ell \approx \sqrt{n}$. Biswal, Lee, and Rao~\cite{BLR10} gave an alternative proof of \Cref{thm:AST} based on spectral graph theory, but the running time is not discussed. In summary, despite intensive research for more than two decades, known algorithms either give a separator of size much larger than $O(\poly(h)\sqrt{n})$ or have superlinear running time, or both. Therefore, \Cref{question:main} remained wide open. See \Cref{table:separator} for a summary.

\begin{table}[!ht]
\small\centering
\renewcommand{\arraystretch}{1.5}
\begin{tabular}{r|r|r|l}
Separator size & Running time & References & Notes
\\\hline
$O(h^{3/2}\sqrt{n})$ & $O(h^{1/2}\sqrt{n}m)$ & \cite{AST90} & \\
$2^{O(h^2)}n^{2/3}$ & $2^{O(h^2)}n$ & \cite{RW09} & $2^{O(h^2)}n^{(2-\eps)/3}$-size in $2^{O(h^2)}n^{1+\eps}$ time $\forall \eps\in [0,1/2]$ \\
$O(h\sqrt{n}+f(h))$ & $O(g(h)\,n^{1+\eps})$ & \cite{KR10} & full proof not available yet \\
$O(h\sqrt{n\log n})$ & $O(h\sqrt{n\log n}m)$ & \cite{PRS94} & setting $\ell = \sqrt{n}/(h\sqrt{\log n})$ in \Cref{eq:shallow} \\
$O(\poly(h)n^{4/5+\eps})$ & $O(\poly(h)\,n)$ & \cite{WulffNilsen11} & fast implementation of~\cite{PRS94} \\
$O(h\sqrt{n\log n})$ & $O(\poly(h)\,n^{5/4+\eps})$ & \cite{WulffNilsen11} & fast implementation of~\cite{PRS94}\\
$O(\poly(h)\sqrt{n}\log^2 n)$ & $O(\poly(h)\,n\polylog(n))$ & \cite{Racke14,Peng16,LS22} & balanced cut, do not output a minor model directly\\
$O(\poly(h)\sqrt{n})$ & $O(\poly(h)\,n)$ & This paper & \Cref{thm:main}\\
\hdashline
$O(\sqrt{n})$ & $O(n)$ & \cite{LT79} & planar graphs\\
$O(\sqrt{gn})$ & $2^{O(\poly(g))}n$ & \cite{GHT84,KMR08} & genus-$g$ graphs\\
$O(\poly(g)\sqrt{n})$ & $O(\poly(g)\,n)$ & This paper & genus-$g$ graphs, direct corollary of \Cref{thm:main} \\
\end{tabular}
\caption{The top part of the table contains separator theorems for $K_h$-minor-free graphs. Functions $f(h)$ and $g(h)$ in line 3 are of the Robertson--Seymour type: a~tower of several exponentials.}
\label{table:separator}
\end{table}

We note that a balanced separator could be obtained by approximating (vertex variants of) sparsest cuts, or more precisely, balanced cuts, in graphs. The cut-matching game~\cite{KRV09} and its non-stop version~\cite{Racke14}, together with recent developments on approximate maxflow~\cite{Sherman13,Peng16} and vertex-expander decomposition~\cite{LS22}, give a~$\polylog(n)$-approximation algorithm in $O(m\polylog(n))$ time. For minor-free graphs, one may assume $m = O(\poly(h)n)$~\cite{Kostochka82}, and therefore, the running time becomes $O(\poly(h)n\polylog(n))$. However, these techniques inherently have multiple log-factor overheads in both the separator size and the running time---in particular, the $\polylog(n)$ in the running time is at least $\log^{40}n$---and hence do not seem to be the right approach for settling \Cref{question:main} completely. Furthermore, they often require sophisticated subroutines, e.g., the cut-matching game or graph sparsifiers. When they fail to construct a small separator, these techniques often output expander-like subgraphs as a failure certificate.  In the context of this paper, when an algorithm fails, we want a  $K_h$-minor model instead. To turn an expander-like subgraph into a  $K_h$-minor model, non-trivial additional effort seems to be needed. Here, a \EMPH{$K_h$-minor model} of a graph $G$ is a collection of $h$ vertex-disjoint connected subgraphs $\{C_1,C_2,\ldots, C_h\}$ of $G$ such that for every two subgraphs $C_i,C_j$ where $i\neq j$, there exists an edge $uv \in E(G)$ where $u\in V(C_i)$ and $v\in V(C_j)$.

\subsection{Our Contribution}\label{subsec:contribution}

In this paper, we answer \Cref{question:main} completely: for a given $K_h$-minor-free graph with $n$ vertices, we construct a~balanced separator of size $O(\sqrt{n})$ in $O(n)$ time for any fixed $h$, and furthermore, the dependence on $h$ in both the separator size and the running time is polynomial. Specifically, in deterministic $O(\poly(h)n)$ time, our algorithm either constructs a separator of size $O(\poly(h)\sqrt{n})$ or outputs $\bot$ to indicate that the input graph $G$ contains a $K_h$ as a minor. Whenever the algorithm fails to construct a balanced separator, if we are required to output a $K_h$-minor model, we can do so in \emph{randomized linear time} $O(\poly(h)n)$ with success probability at least $1/2$. As an added bonus, our algorithm is extremely simple: we find a balanced separator by running breadth-first search (BFS) on a~suitable graph $O(\poly(h))$ times! See \Cref{subsec:idea} for more details. Our results are summarized in the following theorem.

\begin{restatable}{theorem}{MainTheorem}\label{thm:main}
Let $G$ be any given graph with $n$ vertices, and $h > 0$ be a parameter. There is an algorithm that runs in deterministic $O(\poly(h)n)$ time and outputs either:
\begin{itemize}
    \item a balanced separator of size $O(\poly(h)\sqrt{n})$ or
    \item $\bot$ if $G$ contains a $K_h$ as a minor.
\end{itemize}
Furthermore, in the latter case, we can output a $K_h$-minor model of $G$ with probability at least $1/2$ in an additional randomized $O(\poly(h)n)$ time.
\end{restatable}

The precise separator size is $O(h^{13}\sqrt{n})$ and the precise running time is $O(h^{13}n)$. In this paper, we do not attempt to minimize the exponent of $h$ to keep a simple presentation.  When we are not required to output a $K_h$-minor model, since a $K_h$-minor-free graph has $O(h\sqrt{\log h}n)$ edges~\cite{Kostochka82}, we simply output $\bot$ when $m > c\cdot h\sqrt{\log h}n$, where $m \coloneqq |E(G)|$, for a sufficiently large constant $c$ to indicate that $G$ is not $K_h$-minor-free. If we are required to output a~$K_h$-minor model in this case, then we can only keep the first $c\cdot h\sqrt{\log h}n$ (arbitrary) edges from $G$, and find a $K_h$-minor model of the resulting graph. However, existing algorithms~\cite{RW09,DHJRW13}  have running time $2^{\Omega(h)}n$:  \cite{RW09} requires at least $2^{h-3}n$ edges and takes time $2^{O(h)}n$, while \cite{DHJRW13} requires only $\Theta(h\sqrt{ \log h})n$ edges but takes time $\poly(h)n+2^{\Tilde{O}(h^2)}$ due to a brute-force step on a $\Theta(h^2 \log{h})$-vertex graph. There are two ways we can resolve this issue.  We could run our algorithm without the assumption that $m  = O(\poly(h)n)$, and the final running time would be $O(\poly(h)m)$. Alternatively, and more efficiently, we provide an algorithm (\Cref{lem:minors-in-dense}) that in deterministic $O(\poly(h)n)$ time, produces a $K_h$-minor model when $m \geq 100 h^2n$.  
As a result, we can assume that $m = O(h^2 n)$.  Note that we do not have a dependence on $m$ in the running time since we only read the first $100 h^2n$ edges of $G$, and ignore other edges. %

Our \Cref{thm:main} applied to graphs of genus-$g$ also gives a new result: we can find a separator of size $O(\poly(g)\sqrt{n})$ in $O(\poly(g)n)$ time, avoiding computing the surface embedding of the input graph, which is a difficult task, as remarked above. %

\subsection{Technical Ideas}\label{subsec:idea}

For simplicity, we assume that the input graph is $K_h$-minor-free, and our goal is to construct a separator of size $O(\poly(h)\sqrt{n})$ in deterministic linear time $O(\poly(h)n)$. We say that a subset of vertices $S\subseteq V$ is an \EMPH{$\alpha$-balanced separator} of $G = (V,E)$ for some $\alpha \in (0,1)$ if every connected component of $G\setminus S$ has size at most $\alpha\cdot|V|$. We simply say that $S$ is a \EMPH{balanced separator} if it is $2/3$-balanced. Our algorithm is given in \Cref{alg:sep-simplified}; it outputs a~$\bigl(1-\frac{1}{200h^2}\bigr)$-balanced separator of size $O(\poly(h)\sqrt{n})$. To get a $2/3$-balanced separator, we simply repeat the algorithm $O(h^2)$ times and take the union of all the separators from all the runs. The size of the separator increases by a factor of $O(h^2)$ and hence remains $O(\poly(h)\sqrt{n})$. The highlighted lines in \Cref{alg:sep-simplified} are where we apply the vertex-weighted variant of BFS to $G$. These weights are designed by our algorithm, starting from uniform weights for all vertices. We will clarify the role of the weight function shortly.

To formally define our weighted variant of BFS, we need to introduce some notation. Let $w: V \rightarrow \mathbb{Z}_+$ be a positive integer weight function on the vertices of $G$; we write $\angbracket{G,w}$ to denote $G$ with weight function $w$ attached to it. We define the \EMPH{length} of a path $P$ ($w(P)$) in $\angbracket{G,w}$ to be the total weight of the vertices on the path. (Note that if the path contains a single vertex $v$, its length is $w(v)$ instead of $0$.) The (vertex-weighted) \EMPH{distance} between two vertices $u$ and $v$, denoted by $d_{\angbracket{G,w}}(u,v)$, is the length of the shortest path between $u$ and $v$. Edges of graphs in our paper are not weighted, so the length and distance terminologies are meant to be vertex-weighted only.

The procedure \EMPH{$\bfs(\angbracket{G,w},v,r)$} computes a (vertex-weighted) shortest-path tree, say $T$, rooted at $v$ in $\angbracket{G,w}$ truncated at radius $r$. That is,  $\bfs(\angbracket{G,w},v,r)$ only contains vertices at distance at most $r$ from $v$.   Let $W = \sum_{v}w(v)$ be the total vertex weight. By standard reduction, one could reduce computing the vertex-weighted shortest-path tree to computing a BFS tree in a (directed) graph with $\{0,1\}$-edge weights in time $O(m+W)$; see \Cref{lm:node-BFS}. In our algorithm, we will guarantee that $W = \poly(h)n$ and hence the running time of $\bfs(\angbracket{G,w},v,r)$ is $O(m + \poly(h)n)$.

Another key subroutine of our algorithm is the vertex-weighted variant of the \EMPH{KPR decomposition}, named after Klein, Plotkin, and Rao~\cite{KPR93}. The KPR decomposition, given an unweighted\footnote{KPR could work for edge-weighted graphs as well; in our paper, the intuition came from the unweighted version.} $K_h$-minor-free graph $G$ with $m$ edges and a parameter $\Delta > 0$, decomposes $G$ into connected components of \EMPH{weak (hop) diameter} $O(h^2\Delta)$ by removing $O(h\cdot m/\Delta) =  (h^2 \sqrt{\log h} n/\Delta)$ edges. When $\Delta \approx \sqrt{n}$, the number of edges removed is $O(\poly(h)\sqrt{n})$, coinciding with the separator bound that we are aiming for.  The decomposition is obtained by \EMPH{computing BFS trees in $h$ rounds}. Here in our work, we need ($i$) a vertex-weighted version of KPR and ($ii$) to output a $K_{h}$-minor model in linear time whenever the algorithm fails to output the largest component of small (weak) hop diameter. Therefore, our guarantees, summarized in \Cref{lm:KPR} below, are slightly different from the original KPR~\cite{KPR93}. The proof will be given in \Cref{app:KPR-proof}. 

\begin{figure*}[ht!]
\centering%
\begin{algorithm}
\textul{$\textsc{FindSeparator}(G=(V,E))$:}$\qquad$ \Comment{Assume that $G$ is $K_h$-minor-free} \+
\\   $w_1(v)\leftarrow 40$ for every $v\in V$
\\   $S\leftarrow \emptyset\qquad$  \Comment{The separator}  
\\   for $t\leftarrow 1$ to $20h^2$ \+
\\      \hl{$(C^{*}_t, S_t)\leftarrow \kpr(\angbracket{G,w_t},\lfloor\sqrt{n}/12h^2\rfloor,h)$} $\qquad$ \Comment{KPR with $\Delta = \lfloor\sqrt{n}/12h^2\rfloor$, applying BFS $2h)$ times}
\\      $S\leftarrow S\cup S_t$
\\      if $|C^*_t|\leq (1-\frac{1}{200h^2})n$ $\qquad$ \Comment{The current separator is $(1-\frac{1}{200h^2})$-balanced}\+
\\          return $S$    \-  
\\      $c_t\leftarrow$ an arbitrary vertex in $C^*_t$
\\      \hl{$T_{t}\leftarrow\bfs(\angbracket{G,w_t},c_t,\sqrt{n})$} $\qquad$ \Comment{BFS truncated at radius $\sqrt{n}$}
\\      for every $v\in V(T_{t})$ \+
\\          $w_{t+1}(v)\leftarrow w_t(v) + \ceil{\frac{|T_{t}(v)|20^3h^6}{\sqrt{n}}\cdot w_t(v)}  \qquad$ \Comment{Reweighting vertices for next iteration} \-\-
\\   return $S$
\end{algorithm}
\caption{An algorithm for finding a $(1-\frac{1}{200h^2})$-balanced separator with size $O(\poly(h)\sqrt{n})$ of a $K_h$-minor-free graph $G$. $T_{i}(v)$ denotes the vertex set of the subtree of $T_{i}$ rooted at $v$.}
\label{alg:sep-simplified}
\end{figure*}

\begin{restatable}[Vertex-Weighted KPR]{lemma}{VertexWeightKPR}\label{lm:KPR} 
Given a graph $\angbracket{G = (V,E),w}$ with integer weights on vertices, $m$ edges and integer parameters $\Delta > 0, h > 0$. Let $W = \sum_{v\in V}w(v)$. The procedure $\kpr(\angbracket{G,w}, \Delta,h)$ runs in  time $O(h\cdot (m+W))$ and returns either:
\begin{itemize}
    \item a pair $(C^*, S)$ where $S$ is a subset of vertices of size $h\cdot W/\Delta$, and $C^*$ is the connected component of \emph{maximum size} of  $G - S$, which is guaranteed to have (vertex-weighted) weak diameter at most $12\cdot h^2\Delta$, or
    \item  a $K_h$-minor model.
\end{itemize}
\end{restatable}

As we will guarantee that $W = \poly(h)n$, the total running time of $\kpr(\angbracket{G,w}, \Delta)$  is $O(\poly(h)m)$. We will set $\Delta = \lfloor \frac{\sqrt{n}}{12h^2}\rfloor$ (in \Cref{alg:sep-simplified}) and therefore, the number of vertices in the separator returned by KPR is $O(\poly(h)\sqrt{n})$. Similar to the original KPR~\cite{KPR93},
our vertex-weighted version of KPR decomposition can be constructed by applying (vertex-weighted) BFS $2h$ times. Thus, in total, our algorithm calls vertex-weighted BFS $O(h^3)$   times. Therefore, it is both simple, efficient, and completely different from all other existing algorithms!

Our key technical contribution is a weighting scheme on the vertices, which offers a new insight into the dynamics between the size of the balanced separator and the existence of a $K_h$-minor model.  The scheme serves two different purposes: guiding the search for the small balanced separator, and the search for the minor model when such a separator does not exist. (The second purpose does not show up in \Cref{alg:sep-simplified}; instead, it will be part of the analysis and in the algorithm for constructing a minor model if required.)

We first describe how the weighting scheme guides the search for a balanced separator. Initially, every vertex $v$ has the same role and hence is given the same weight $w_1(v) = 40$. When we apply KPR to find a separator $S_t$ (at an iteration $t$), the largest component $C^*_t$ of $G-S_{t}$ is only guaranteed to have a small diameter (of $O(\sqrt{n})$); it could contain almost all the vertices.  If in the next iteration, we simply apply KPR again on the large component, we might not get anything, i.e., the returned separator is empty. By increasing the current weight of a vertex  $v$ to:%
\begin{equation}\label{eq:reweight}
  w_{t+1}(v)\leftarrow w_t(v) + \ceil*{\frac{|T_{t}(v)|20^3h^6}{\sqrt{n}}\cdot w_t(v)} \approx w_t(v)\left(1 + \frac{|T_{t}(v)|20^3h^6}{\sqrt{n}}\right),
\end{equation}
where $T_t(v)$ is the vertex set of the \EMPH{subtree rooted at $v$} and $|T_t(v)|$ is its (unweighted) \EMPH{size}, the (vertex-weighted) diameter of the component in the next iteration increases. The ceiling in \Cref{eq:reweight} is to guarantee that $w_{t+1}(v)$ is an integer. (In \Cref{alg:sep-simplified}, we apply BFS to $G$ in every iteration, but we should think of this as applying BFS to the largest component from the previous iteration.)  Vertices that have many descendants in $T_t$ have more weights, and therefore, are more likely to be added to the separator in the next iteration.  Intuitively, vertices with more descendants are responsible for connecting many other vertices, and therefore should be in the separator.  A simplistic yet illuminating example is the star graph: the center of the star will be assigned $\Omega(\sqrt{n})$ weight, which is larger than $12h^2\Delta$, and hence will be added to the separator by KPR in the next iteration. (Note that in general, removing a single or $O(1)$ vertices in a graph of diameter $O(\sqrt{n})$ does not produce small components, e.g., the $\sqrt{n}\times \sqrt{n}$ planar grid.) If we only want to take vertices with many descendants to the separator, then we can simply choose a threshold on the number of descendants, say $\sqrt{n}$, and add a vertex to the separator if the number of its descendants crosses the threshold. Simple thresholding does not work since there is no basis for stopping after $O(h^2)$ iterations of KPR. The second purpose of our weighting scheme, described next, provides such a basis, limiting the number of iterations to $20h^2$. 

Suppose that we want to certify that our algorithm fails by constructing a $K_h$-minor model. We construct a (random) function $\phi$ that maps each vertex $i\in [h]$  of $K_h$ to a (unique) vertex $\phi(i)\in V$ and each edge $(i,j)\in K_h$ to a path, denoted by $\phi(i,j)$, in $G$. Two paths $\phi(i,j)$ and $\phi(a,b)$ have to be vertex-disjoint whenever $\{i,j\}\cap \{a,b\}=\emptyset$. For each vertex $i\in K_h$, $\phi(i)$ is a vertex in $V$ chosen uniformly at random. To find a path $\phi(i,j)$ for every edge $(i,j)\in K_h$, we rely on the set of trees $\mathcal{T} = \{T_1,T_2,\ldots, T_{h^2}\}$, which is produced by the algorithm\footnote{The algorithm produces $20h^2$ trees, and the constant $20$ is needed for the probabilistic analysis; we assume a constant $1$ for simplicity.} in \Cref{alg:sep-simplified} after all iterations. Specifically, edge $(i,j) \in K_h$ is assigned a unique tree $T_{ij}\in \mathcal{T}$, and the path $\phi(i,j)$ is the (only) path between\footnote{$T_{ij}$ might not include all the vertices of $G$. But it does include at least $(1-1/(10h^2))n$ vertices from $G$, which is the size of the largest component. Therefore, our probabilistic analysis remains feasible, albeit somewhat more complex.} $\phi(i)$ and $\phi(j)$ in $T_{ij}$. To obtain a  $K_h$-minor model from $\phi$, it must have the guarantee that whenever  $\{i,j\}\cap \{a,b\}=\emptyset$, the probability of  $\phi(i,j)\cap \phi(a,b) = \emptyset$ is high, say $1-\frac{1}{h^2}$, for the union bound to work.  Suppose that $T_{ij}$ is added to $\mathcal{T}$ before $T_{ab}$. If the subtree $T_{ij}(v)$ rooted at $v$ has a large size, say $n/2$ vertices, then with probability $1/4$, $\phi(i,j)$ will contain $v$ since $\phi$ maps vertices to $V$ uniformly at random. If in $T_{ab}$, $T_{ab}(v)$ also has a large size, then $\phi(a,b)$ will also contain $v$ with a good probability, say $1/4$. Thus, in this thought experiment, with a probability $1/16$, $\phi(i,j)\cap \phi(a,b) \not= \emptyset$.  (Recall we want this probability to be at most $1/h^2$.) By reweighting $w_t(v)$ after the construction of $T_{ij}$ in iteration $t$, $w_{t+1}(v)$ will be substantially larger than $w_t(v)$. Thus, in future iterations, if $v$ has accumulated too much weight, it will be added to the separator by KPR. Otherwise, $v$ has not accumulated too much weight, then in future trees, $v$ will have very few descendants, which directly translates to a low collision probability between $\phi(i,j)$ and $\phi(a,b)$. We formalize this intuition via the notion of a \EMPH{stochastic connector} (\Cref{def:stoc-conn}). The precise weight update rule in \Cref{eq:reweight} allows us to construct a stochastic connector, which in turn gives a $K_h$-minor model, whenever our algorithm fails to return a small separator. Given that balanced separators are deeply connected to sparsest cuts, we believe that the notion of a stochastic connector may be of broader interest.

\subsection{Other Related Work}

An orthogonal research direction motivated by the work of Alon, Seymour, and Thomas~\cite{AST90Conf,AST90} is to minimize the dependency on $h$ in the size of the separator. Recall that their separator has size $O(h^{3/2} \sqrt{n})$. The aforementioned result (\Cref{eq:shallow}) by Plotkin, Rao, and Smith~\cite{PRS94} implies a separator of size $O(h\sqrt{\log n}\sqrt{n})$, which is an improvement over $O(h^{3/2} \sqrt{n})$ whenever $h\gg \log(n)$. Kawarabayashi and Reed~\cite{KR10,Kawara11} claimed to improve the bound to $O(h\sqrt{n} + f(h))$ by following the deep structure decomposition of minor-free graphs in the graph minor theory of Robertson and Seymour. Recently, Spalding-Jamieson~\cite{SpJ25} gave a separator of size $O(h\polylog(h)\sqrt{n})$ based on a reweighted spectral partitioning technique. The algorithm uses  complicated subroutines, e.g., semidefinite programming and padded decomposition~\cite{CF25}. This approach has a $\Omega(n\cdot m)$ barrier, which seems to require $\Omega(n)$ iterations of the algorithm, which takes $\Omega(m)$ time each.

\section{Stochastic Connector}

A central concept in the analysis of our algorithm is that of a~\EMPH{stochastic connector}, which is essentially a collection of rooted trees that have a certain probabilistic guarantee on the random paths sampled from them. We can show that if a stochastic connector of large size (having many trees) exists, then we can construct a clique minor of large size.

Let $T$ be a tree and $u,v$ be two vertices in $T$. We denote by $T[a,b]$ \EMPH{the (unique) path between $a$ and $b$ in $T$}. Suppose that $T$ is rooted at a vertex $r$. We denote by $\mathcal{R}(T) = \{T[r,v]: v\in V(T)\}$ the set of all rooted paths in $T$. By $P\sim \mathcal{R}(T)$, we denote a path $P$ sampled uniformly at random, or \emph{u.a.r}, from $\mathcal{R}(T)$. The sampling of $P$ can be realized in linear time by sampling a vertex $v\in V(T)$ u.a.r and returning the path $T[r,v]$.

\begin{definition}[Stochastic Connector]\label{def:stoc-conn} A \EMPH{stochastic connector of size $k$} for a given graph $G=(V,E)$ is a set of $k$ trees   $\mathcal{T} = \{T_1,T_2,\ldots, T_k\}$ such that:
\begin{enumerate}
    \item $|V(T_i)| \geq n-n/(10k)$ for every $i\in [k]$.
    \item  For every $i,j\in [k]$ such that $i\not= j$:
    \begin{equation}\label{eq:collision}
         \Pr_{P\sim \mathcal{R}(T_i), Q\sim \mathcal{R}(T_j)}[P\cap Q \not=\emptyset] \leq \frac{1}{5k^2}.
    \end{equation}
\end{enumerate}
\end{definition}
The constants $10$ and $5$ in \Cref{def:stoc-conn} look somewhat arbitrary. 
We choose these constants for the union bound to work, and they better fit our algorithm given in \Cref{alg:sep-simplified}. Here, we note that a tree in $\mathcal{T}$ might not contain all vertices of $G$.

Let $H$ be a graph. Let $\mathcal{P}(G)$ be the set of all simple paths in $G$.  An \EMPH{almost-embedding} of $H$ to $G$ is a function $\phi$ that maps $V(H)\rightarrow V(G)$ and $E(H)\rightarrow \mathcal{P}(G)$ such that:
\begin{enumerate}
    \item  for every $e = uv \in E(H)$, $\phi(e)$ is a $\phi(u)$-to-$\phi(v)$ path.
    \item  for any two edges $e_1\not= e_2$ of $H$ that do not share an endpoint, $\phi(e_1)\cap \phi(e_2) = \emptyset$.
\end{enumerate}

An object similar to an almost-embedding was considered by Korhonen and Lokshtanov~\cite{KL24}.
Their techniques show that one can construct a $K_t$-minor model from an almost embedding (into $G$) of a graph of sufficient size.
We capture this in the following lemma, whose proof follows directly from the ideas of~\cite{KL24}, but for completeness we include a proof in \Cref{sec:subcubic}.

\begin{restatable}[{Corresponds to Lemma 4.3 in~\cite{KL24}}]{lemma}{almostembeddinglem}\label{lm:subcubic} Let $G$ be a connected graph with $m$ edges. For every $t\geq 1$, there is a graph $H$ with $|V(H)| + |E(H)| \leq 3t^2$ such that an almost-embedding of $H$ can be turned into a $K_t$-minor model of $G$ in $O(t^2m)$ time.
\end{restatable}

The key result of this section is the following lemma, showing that the existence of a large stochastic connector implies the existence of a large clique minor.

\begin{lemma}\label{lm:connector-to-minor} Given a stochastic connector $\mathcal{T}$ of size $k$ in a connected graph $G$ with $m$ edges, for any $t$ such that $20t^2\leq k$, we can construct a  $K_t$-minor model of $G$ in $O(km)$ time with probability at least $2/5$.
\end{lemma}
\begin{proof}
    By \Cref{lm:subcubic}, it suffices to construct an almost-embedding $\phi$ of any graph $H$ with  $|V(H)| + |E(H)| \leq 20t^2$ into $G$ whenever $k\geq 20t^2$.  Our construction is randomized and as follows:
    \begin{enumerate}
        
        \item  For every vertex $v\in V(H)$, we assign $\phi(v)$ to be a vertex in $V(G)$ chosen u.a.r.
        \item For each edge  $e = uv \in  E(H)$,  we associate $e$ with a distinct tree in $\mathcal{T}$, denoted by $T_e$. This is possible since $k\geq |E(H)|$. Then we assign $\phi(e)$ to be the path $T_e[u,v]$ in $T_e$. (It could be that $u$ or $v$ is not in $T_e$,  and in this case, we set $\phi(e) = \emptyset$.)
    \end{enumerate}

Now we bound the probability that $\phi$ is valid.  We say that two edges are \EMPH{adjacent} if they share an endpoint, and \EMPH{non-adjacent} otherwise. We say that $\phi$ is \EMPH{invalid} when at least one of the following events happens:
\begin{itemize}
    \item \textbf{Event $A$:}  there exists $e = uv \in E(H)$ such that $\{\phi(u),\phi(v)\}\not\subseteq V(T_e)$. 
    \item \textbf{Event $B$:} there exist two non-adjacent $e_1,e_2$ such that, conditioning on $\Bar{A}$,  $\phi(e_1)\cap \phi(e_2)\not=\emptyset$.  
\end{itemize}

     By the first property of $\mathcal{T}$ in \Cref{def:stoc-conn}, for any tree $T\in \mathcal{T}$ and any vertex $v\in H$, $\pr[\phi(v)\not\in V(T)]\leq 1/(10k)$. Thus, for every edge $uv$,  by the union bound, $\Pr[\{\phi(u),\phi(v)\}\not\subseteq V(T_e)] \leq 2/(10k) = 1/(5k)$. Then by union bound over at most $20t^2\leq k$ edges of $H$, we have:
      \begin{equation}\label{eq:PrA}
          \Pr[A]\leq k\cdot 1/(5k)\leq 1/5.
      \end{equation}

     To bound $\pr[B]$, for every edge $e = uv$, let $\hat{\phi}(e) = T_e[u,r_e]\cup T_e[r_e,v]$ where  $r_e$ is the root of $T_e$. Then for any two non-adjacent edges $e_1 = (u_1,v_1)$ and $e_2 = (u_2,v_2)$, we have:
     \begin{equation}
     \begin{split}
        \pr[\phi(e_1)\cap \phi(e_2) \not= \emptyset] &\leq \pr[\hat{\phi}(e_1)\cap \hat{\phi}(e_2) \not= \emptyset] \\
         &\leq \sum_{x\in\{u_1,v_1\}}\sum_{y\in \{u_2,v_2\}} \pr[T_{e_1}[x,r_{e_1}]\cap T_{e_2}[y,r_{e_2}] \not= \emptyset]\\
         &\leq \frac{4}{5k^2}. \qquad \text{(by union bound and \Cref{eq:collision})}
     \end{split}
     \end{equation}
     Thus, by taking the union bound over all non-adjacent pairs, we have:
     \begin{equation}\label{eq:PrB}
         \pr[B]\leq \frac{k^2}{2}\cdot \frac{4}{5k^2} \leq 2/5.
     \end{equation} 

     By union bound, and \Cref{eq:PrA} and \Cref{eq:PrB}, we get:
     \begin{equation*}
     \pr[\phi \text{ is invalid}] \leq  \Pr[A] + \Pr[B] \leq 3/5. 
     \end{equation*}
     
    Thus, with probability at least $2/5$, $\phi$ is a valid almost-embedding. Since finding a path in each tree in $\mathcal{T}$ takes $O(kn)$ time, by \Cref{lm:subcubic}, we can construct a $K_t$-minor model in time $O(k(n+m)) = O(km)$. 
\end{proof}

\section{Constructing Balanced Separators}\label{sec:balanced-sep}

In this section, we prove \Cref{thm:main}, which we restate below for convenience.
\MainTheorem*

The algorithm for constructing a balanced separator in \Cref{alg:sep-simplified} assumes that the input graph is $K_h$-minor-free. Here in \Cref{thm:main}, we do not make this assumption; instead, we want to output a $K_h$-minor when we fail to find a good balanced separator. We do so by making two simple changes to \Cref{alg:sep-simplified}. First, the KPR decomposition in every iteration $t$ returns a tuple $(C^*_t,S_t,M_t)$ where $M_t$ is a $K_h$-minor model whenever it fails to find a small separator $S_t$ such that the largest component $C^*_t$ has small weak diameter. By \Cref{lm:KPR}, either $M_t= \emptyset$---in this case, the KPR returns the pair $(C^*_t,S_t)$---or $M_t\not= \emptyset$, and in this case  $M_t$ is a $K_h$-minor model. Second, instead of returning a separator at the end of the algorithm\footnote{If $G$ is $K_h$-minor-free, then the analysis in this section shows that the algorithm will never reach this line.} as in \Cref{alg:sep-simplified}, we run the algorithm  $\textsc{FindMinor}(G, \{T_1,T_2,\ldots, T_{k}\})$ with $k=20 h^2$ to construct a $K_h$-minor by \Cref{lm:connector-to-minor}. Note that \Cref{lm:connector-to-minor} requires the set of trees $\{T_1,T_2,\ldots, T_{k}\}$ to be a stochastic connector. Our analysis in this section shows that this indeed is the case.   Finally, as noted in \Cref{subsec:contribution}, to reduce the running time from $O(\poly(h)m)$ to $O(\poly(h)n)$, we construct a $K_h$-minor model directly as formally stated in the following lemma, whose proof will be given in \Cref{sec:sparse-reduction}.

\begin{figure*}[ht!]
\centering%
\begin{algorithm}
\textul{$\textsc{FindSeparator}(G=(V,E))$:}\+
\\      if $|E|\geq 100h^2 n$ \hspace{1cm} \Comment{$G$ is dense}\+
\\         return $K_h$-minor model given by \Cref{lem:minors-in-dense}\-
\\   $w_1(v)\leftarrow 40$ for every $v\in V$
\\   $S\leftarrow \emptyset\qquad$ \hspace{1cm} \Comment{The separator}  
\\   $k\leftarrow 20h^2$ 
\\   for $t\leftarrow 1$ to $k$  \hspace{1cm} \Comment{$O(h^2)$ iterations}  \+
\\      \hl{$(C^*_t, S_t, M_t)\leftarrow \kpr(\angbracket{G,w_t},\lfloor\sqrt{n}/12h^2\rfloor,h)$} $\qquad$ \Comment{KPR with $\Delta = \lfloor\sqrt{n}/12h^2\rfloor$, applying weighted BFS $2h$ times}
\\      if $M_t\not=\emptyset$ \hspace{1cm} \Comment{$K_h$-minor model found}\+
\\          return $M_t$  \-
\\      $S\leftarrow S\cup S_t$
\\      if $ |C^*_t|\leq (1-\frac{1}{10k})n$ $\qquad$ \Comment{The current separator is $(1-\frac{1}{10k})$-balanced}\+
\\          return $S$    \-  
\\      $c_t\leftarrow$ an arbitrary vertex in $C^*_t$
\\      $T_{t}\leftarrow\bfs(\angbracket{G,w_t},c_t,\sqrt{n})$ $\qquad$ \Comment{BFS truncated at radius $\sqrt{n}$}
\\      for every $v\in V(T_{t})$ \+
\\        $w_{t+1}(v)\leftarrow w_t(v) + \ceil{\frac{|T_{t}(v)|k^3}{\sqrt{n}}\cdot w_t(v)}  \qquad$\Comment{Reweighting vertices for next iteration} \-\-
\\   return $\textsc{FindMinor}(G, \{T_1,T_2,\ldots, T_{k}\})$ \hspace{1cm} \Comment{Find $K_{h}$-minor model  by \Cref{lm:connector-to-minor}}
\end{algorithm}
\caption{An algorithm that finds a $(1-\frac{1}{200h^2})$-balanced separator with size $O(\poly(h)\sqrt{n})$ of $G$ if exists, or outputs a $K_h$-minor model otherwise.}
\label{alg:full}
\end{figure*}

\begin{restatable}{lemma}{SparseReduction}\label{lem:minors-in-dense}
There is a~deterministic algorithm that, given any $n$-vertex graph $G$ with at least $100 h^2 n$ edges, reads the first $100h^2n$ edges of $G$ and outputs a~$K_h$-minor model of~$G$ in $O(h^6 n)$ time.
\end{restatable}

Now we have all the ideas needed for proving \Cref{thm:main}.

\begin{namedproof}[Proof of \Cref{thm:main}] 
By \Cref{lem:minors-in-dense} we can assume that $m= O(h^2 n)$ since we can find a $K_h$-minor model in $O(h^6n)$ time whenever $m \geq 100h^2n$. 

If the algorithm returns a separator $S$, in~\Cref{lm:sep-size} of~\Cref{sec:correctness}, we will show that $|S| = O(h^{11}\sqrt{n})$. Otherwise, we will show in~\Cref{lm:connector} of~\Cref{sec:correctness} that the set of trees $\{T_1, T_2,\ldots, T_k\}$ found by the algorithm is a~stochastic connector. Thus, the algorithm will return a $K_{h}$-minor model either in one of the iterations or by calling $\textsc{FindMinor}(G, \{T_1,T_2,\ldots, T_{k}\})$. 

If we are not required to output a minor, then in the last line of the algorithm, we return $\bot$ instead of calling $\textsc{FindMinor}(G, \{T_1,T_2,\ldots, T_{k}\})$. The running time, as we will show in \Cref{lm:time}, is $O(h^3m + h^{11}n) = O(h^{11}n)$. 

If we are required to output a minor, then by \Cref{lm:connector-to-minor}, the total running time is:
\begin{equation*}
    O(h^3m + h^{11}n) + O(h^2m) = O(h^3m +  h^{11}n) = O(h^{11} n).
\end{equation*}
Also by \Cref{lm:connector-to-minor}, the algorithm is randomized and succeeds with a probability of at least $2/5$.

Since the separator output by \Cref{alg:full} is only $(1-\frac{1}{200h^2})$-balanced, we need to run the algorithm $O(h^2)$ times to get a $2/3$-balanced separator. The total separator size and running time are $|S| = O(h^{13}\sqrt{n})$ and $ O(h^{13}n)$, respectively.  
\hfill$\square$\null
\end{namedproof}

\subsection{Correctness}\label{sec:correctness}

Here, we only consider weight functions on vertices, not on edges. If $\angbracket{G,w}$ is a graph with a weight function $w$, we say that $X\subseteq V(G)$ has weak diameter at most $D$ if $d_{\angbracket{G,w}}(u,v) \leq D$ for every $u,v\in X$.

Let $E_{t,v}$ be the event that a random rooted path  $P\sim \mathcal{R}(T_t)$ contains $v$, where $T_t$ is the tree at iteration $t$.  We first show that the algorithm maintains five invariants.

\begin{lemma}\label{lm:invariant} \Cref{alg:full} maintains the following invariants at every iteration $t$ if it does not return a $K_h$-minor model or a separator $S$:
\begin{enumerate}[label=\textbullet, ref=\arabic*]
    \item\label{it:1} \textbf{Invariant 1:} $\Pr[E_{i,v}]\leq (20k^3\,\sqrt{n})^{-1}\,w_{t+1}(v)$ for every $i \leq t$ and every $v$.
    \item\label{it:2}  \textbf{Invariant 2:} $\Pr[E_{i,v}\cap E_{j,v}] \leq (10k^6\,n)^{-1}\,w_{t+1}(v)$ for every $i\not= j \leq t$ and every $v$.
    \item\label{it:3}   \textbf{Invariant 3:} $\sum_{v\in V}w_t(v)\leq (40 + (k^3 + 1)(t-1))n$.
    \item\label{it:4}  \textbf{Invariant 4:} $|V(T_t)|\geq (1-\frac{1}{10k})n$.
    \item\label{it:5}  \textbf{Invariant 5:} $w_t(v)\geq 40$ for every $v$.
\end{enumerate}
\end{lemma}
\begin{proof} Initially, $t = 1$ and the algorithm sets $w_1(v) = 40$ for every $v$. Since $w_t(v)$ only increases as $t$ increases, Invariant~\ref{it:5} holds.  Herein, we focus on four other invariants. %

 For Invariant~\ref{it:4}, we first note that since the algorithm does not return a $K_{h}$-minor model or a separator $S$ at iteration $t$, it must be that:
\begin{itemize}
    \item The largest connected component $C^*_t$ of $G - S_t$ has size at least $(1-1/{10k})n$.
    \item $C^*_t$ has weak (vertex-weighted) diameter $12h^2\Delta\leq \sqrt{n}$ by \Cref{lm:KPR}. 
\end{itemize}
Thus, for any vertex $c_t\in V(C^*)$, every other vertex $v\in V(C^*_t)$, $d_{\angbracket{G,w_t}}(c_t,v)\leq \sqrt{n}$. Thus, the BFS tree rooted at $c_t$ truncated at radius $\sqrt{n}$ contains all vertices of $C^*_t$. Thus, we have:
\begin{equation}\label{eq:size-T_t}
    |V(T_t)|\geq |V(C^*_t)|\geq \left(1-\frac{1}{10k}\right)n~,
\end{equation}
giving Invariant~\ref{it:4}. (The BFS tree might contain vertices not in $C^*_t$, and their weights get updated as well, but this has no effect on the algorithm.)

To show Invariant~\ref{it:3}, let $W_t = \sum_{v\in V}w_t(v)$. By induction, we assume that $W_t\leq  (40 + (k^3+1)(t-1))n$. We now show that $W_{t+1}\leq W_t+(k^3+1)n \leq (40 + (k^3+1)t)n$ by a simple charging scheme. Observe that, by the definition of $w_{t+1}(v)$:
\begin{equation}\label{eq:reweight-ceil}
  w_{t+1}(v)\leq  w_t(v)\left(1 + \frac{|T_{t}(v)|k^3}{\sqrt{n}}\right) + 1,
\end{equation}

For each vertex $v\in V(T_{t})$, the algorithm add at most $\frac{|T_{t}(v)|k^3}{\sqrt{n}}\cdot w_t(v) + 1$  to $w_t(v)$. We will charge an amount of $(k^3/\sqrt{n}) w_t(v)$ to each vertex $u\in T_{t}(v)$. Since  $u$ is only charged by its ancestors, the total charge of each vertex $u\in V(G)$ at iteration $t$ is:
\begin{equation*}
    \sum_{v  \text{ ancestor of }u \text{ in }T_t} \frac{w_t(v) k^3}{\sqrt{n}}  \leq \frac{w_t(T_{t}[u,c_{t}])k^3}{\sqrt{n}} \leq k^3
\end{equation*}
since $T_t$ has (vertex-weighted) radius at most $\sqrt{n}$. Thus,  the total weight increase of all vertices at iteration $t$ is at most $n\cdot k^3 + n$ (including the $+1$ for each vertex in \Cref{eq:reweight-ceil}), giving Invariant~\ref{it:3}.

For Invariant~\ref{it:1}, we note that $|V(T_t)|\geq n/2$ by Invariant~\ref{it:4}. Furthermore, by the weight update rule, we have:
\begin{equation}\label{eq:T-bound}
    w_{t+1}(v) \geq \frac{|T_{t}(v)|k^3}{\sqrt{n}} w_t(v) \Rightarrow |T_t(v)|\leq \frac{\sqrt{n}\,w_{t+1}(v)}{k^3\,w_t(v)}.
\end{equation}
Now we observe that if $v \notin V(T_t)$, then $\Pr[E_{t,v}]=0$, and otherwise:
\begin{equation}
\begin{split}
     \Pr[E_{t,v}] &= \frac{|T_{t}(v)|}{|V(T_t)|}  \leq  \frac{2|T_{t}(v)|}{n} \qquad \text{(since $|V(T_t)|\geq n/2$)}\\
     &\leq   \frac{2\sqrt{n}\,w_{t+1}(v)}{k^3 n\,w_t(v)} \qquad \text{(by \Cref{eq:T-bound})}\\
     &\leq \frac{w_{t+1}(v)}{20k^3\sqrt{n}}.  \qquad \text{(since $w_t(v)\geq 40$ by Invariant~\ref{it:5})}
\end{split}
\end{equation}
By induction, for any $i\leq t-1$, we have \[\Pr[E_{i,v}]\leq \frac{w_{t}(v)}{20k^3\sqrt{n}} \leq \frac{w_{t+1}(v)}{20k^3\sqrt{n}}\] since $w_t(v)$ only increases as $t$ increases.  Thus, Invariant~\ref{it:1} holds.

It remains to show Invariant~\ref{it:2}. Observe that $E_{i,v}$ and $E_{t,v}$ are independent for any $i \leq t-1$. Thus, we have:

\begin{equation*}
\begin{split}
    \pr[E_{i,v}\cap E_{t,v}] &= \Pr[E_{i,v}] \Pr[E_{t,v}]  \leq \frac{w_t(v)}{20 k^3\sqrt{n}}\frac{|T_{t}(v)|}{|V(T_{t})|} \qquad \text{(by Invariant~\ref{it:1} and $i\leq t-1$)} \\
&\leq \frac{w_t(v)}{20 k^3\sqrt{n}} \frac{2\sqrt{n}\,w_{t+1}(v)}{k^3 n\,w_t(v)} \qquad \text{(by \Cref{eq:T-bound} and $|V(T_t)|\geq n/2$)}\\
&\leq \frac{w_{t+1}(v)}{10k^6n},
\end{split}
\end{equation*}
giving Invariant~\ref{it:2}. 
\end{proof}

Next, we bound the size of the separator if the algorithm returns one.

\begin{lemma}\label{lm:sep-size} If~\Cref{alg:full} returns a separator $S$, then $|S|\leq O(h^{11}\sqrt{n})$.
\end{lemma}
\begin{proof}
Let $W_t = \sum_{v\in V}w_t(v)$. By \Cref{lm:KPR}, we have:
\begin{equation}
\begin{split}
    |S_t| &=  O(h W_t/(\sqrt{n}/h^2)) \qquad \text{(since $\Delta = \floor{\sqrt{n}/(12h^2)}$)}\\
    &=  O(h^3 k^3 tn/\sqrt{n}) \qquad \text{(by Invariant~\ref{it:3} in \Cref{lm:invariant})}\\
    &  \leq O(h^3 k^4 \sqrt{n})  \qquad \text{(since $t\leq k$)}\\
    &= O(h^{11}\sqrt{n}), \qquad \text{(since $k = O(h^2)$)}
\end{split}
\end{equation}
as claimed.
\end{proof}

Finally, we show that:

\begin{lemma}\label{lm:connector} Suppose that the algorithm does not return a separator $S$ after $k$ iterations, then $\mathcal{T} \coloneqq\{T_1,\ldots, T_k\}$ is a~stochastic connector.
\end{lemma}
\begin{proof}
We verify all the properties of $\mathcal{T}$ in \Cref{def:stoc-conn}. Observe that the size bound follows directly from Invariant~\ref{it:4} in \Cref{lm:invariant}. It remains to bound the collision probability of two random paths by $1/(5k^2)$. Consider two trees $T_i$ and $T_t$ where $i < t$ formed in iteration $i$ and $t$, respectively.  By Invariant~\ref{it:3}, we have:
\begin{equation}\label{eq:total-weight}
    \sum_{v\in V}w_{t+1}(v) \leq (40+(k^3+1)t)n \leq 2k^4 n
\end{equation}
since $t\leq k$ and $k^4 = 20^4 h^8 \geq 40 + k$ as $h\geq 1$.  By taking the union bound  from the probability in Invariant~\ref{it:2}, we have:

\begin{equation*}
       \Pr_{P\sim \mathcal{R}(T_i), Q\sim \mathcal{R}(T_t)}[P\cap Q\not=\emptyset]  \leq  \sum_{v\in V} \frac{w_{t+1}(v)}{10k^6n}  \stackrel{\text{\cref{eq:total-weight}}}{\leq}  \frac{2k^4n}{10k^6n}= \frac{1}{5k^2},
\end{equation*}
as desired. 
\end{proof}

\subsection{Running time}

We first bound the running time of (vertex-weighted) BFS.

\begin{lemma}\label{lm:node-BFS}
    Let $\angbracket{G = (V,E),w}$ be a connected vertex-weighted graph with $m$ edges. Suppose that the vertex weights are integers. Let $W = \sum_{v\in V}w(v)$. In $O(m+W)$ time, we can compute a node-weighted shortest-path tree rooted at any vertex $c\in V$. 
\end{lemma}
\begin{proof}
    First, we apply the textbook reduction (for maxflow with vertex capacities) to reduce to computing the shortest-path tree in a~directed graph with edge weights in  $\{0,1\}$: split every vertex $v$ into $v_{in}$ and $v_{out}$, add a~directed path from  $v_{in}$ to $v_{out}$ of total length $w(v)$ (so each edge has length 1), and for every edge $uv \in E(G)$, we add two arcs $(v_{out}\rightarrow u_{in})$ and $(u_{out}\rightarrow v_{in})$ of length $0$ each.  Let $\hat{G}$ be the resulting graph. Observe that $\hat{G}$ has $O(m + W)$ directed edges. Now we can find a shortest-path tree of $\hat{G}$ from $c_{in}$ by doing a~breadth-first search in time $O(m+W)$. Then in the same time, we can convert the BFS tree of $\hat{G}$ into a node-weighted shortest-path tree rooted at $c\in V$. 
\end{proof}

Now we bound the running time of~\Cref{alg:full}.

\begin{lemma}\label{lm:time} The running time of~\Cref{alg:full}, excluding the running time of \textsc{FindMinor} (the last line), is $O(h^3m + h^{11}n)$
\end{lemma}
\begin{proof} 
Let $W_t = \sum_{t}w_t(v)$.  By \Cref{lm:KPR}, the running time of the KPR decomposition in iteration $t$ is  $O(h(m+W_t))$. The running time to compute the (vertex-weighted) BFS is $O(m+W_t)$ by \Cref{lm:node-BFS}. Thus, the total running time is:
\begin{equation}
    \sum_{t=1}^{k}O(h(m + W_t)) = O(hkm + h\sum_{t=1}^k W_t) = O(hkm + h k^5 n) = O(h^3m + h^{11}n), 
\end{equation}
as desired.
\end{proof}

\section{Reduction to Sparse Graphs}\label{sec:sparse-reduction}

The goal of this section is to show \Cref{lem:minors-in-dense}, restated below, in order to avoid any dependence on the number $m$ of edges of~\Cref{thm:main}.
\SparseReduction*

To that end, we could combine the results in \cite{DHJRW13} and \cite{AlonKS23}.
Dujmović, Harvey, Joret, Reed, and Wood~\cite{DHJRW13} describe an algorithm that finds in $O(\poly(h)n + 2^{\Tilde{O}(h^2)})$ time a~$K_h$-minor model in any $n$-vertex graph with average degree at least $10 h \sqrt{\log h}$.
Apart from the final brute-force search on an~$O(h^2 \log h)$-vertex graph, their algorithm runs in time $O(\poly(h)n)$.
As this graph still has large enough average degree, one can thus substitute this last step with an effective version of the short proof by Alon, Seymour, and Thomas~\cite{AlonKS23} of the Kostochka--Thomason's bound.
One can indeed easily derive a~randomized polynomial-time algorithm from their proof.

Instead, we will give a~simple proof with a~worst (quadratic) bound on the average degree.
This way our algorithm is deterministic and self-contained.
The first step is to efficiently find, in any graph of average degree~$2d$, a~minor with minimum degree at~least~$d$ and only $\Theta(d)$ vertices.

We denote the minimum degree of~$G$ by $\delta(G)$.
Given a~graph $H$ with vertex set $\{v_1, \ldots, v_h\}$, an~\EMPH{$H$-minor model} of a graph $G$ is a collection of $h$ vertex-disjoint connected subgraphs $\{C_1,C_2,\ldots, C_h\}$ of $G$ such that for every two subgraphs $C_i,C_j$ when $(v_i,v_j) \in E(H)$, there exists an edge $uv \in E(G)$ with $u\in V(C_i)$ and $v\in V(C_j)$.

\begin{lemma}\label{lem:densifier}
For any positive integer $d$, there is an algorithm that on any $n$-vertex graph $G$ with $|E(G)|=dn$ outputs an $H$-minor model of $G$ with $|V(H)| \leq 2d$ and $\delta(H) \geq d$, and runs in time $O(d^3 n)$.
\end{lemma}

\begin{proof}
The algorithm goes through a~strictly decreasing chain of graphs for the minor relation. 
In particular, it builds at~most~$|E(G)|+|V(G)|=(d+1)n$ (distinct) graphs.
In addition to the current graph $G'$, it maintains 
\begin{itemize} 
\item an array $t$ of length~$n$ such that $t[i]$ contains a~doubly linked list of vertices with degree exactly $i$ in $G'$, 
\item an array $s$ of length $n$ such that for any $v \in V(G)$, $s[v]$ is a~pointer to the occurrence of $v$ in $t$, or $\bot$ if $v \notin V(G')$, 
\item $\delta := \delta(G')$ the minimum degree of~$G'$, and
\item an array $f$ of length $n$ such that for any $v \in V(G')$, $f[v]$ contains the list of original vertices of~$G$ contracted into~$v$.
\end{itemize}
From a~representation of~$G$ by adjacency lists, $s$, $t$, $f$ and $\delta$ can be initialized (for $G'=G$) in time $O(dn)$.
For $G'$, we use a~data structure that supports edge deletion and edge addition in $O(1)$ time, and deletion of vertex $v \in V(G')$ in time $O(d_{G'}(v))$.
We define the positive measure $\mu(G') := 2|E(G')|+|V(G')|+(n-\delta(G')) \leq (2d+2)n$.

We iteratively apply the following steps:
\begin{enumerate}
    \item Let $u$  be the first vertex in $t[\delta]$.
\item If $\delta < d$, simply delete $u$ from $G'$,
\item else if $\delta > 2d$, delete all but $2d$ edges incident to~$u$ (arbitrarily),
\item else if there is an edge $uv \in E(G')$ such that $|N_{G'}(u) \cap N_{G'}(v)| < d$, then contract $uv$, by deleting $u$ and adding the edges $vw$ for every $w \in N_{G'}(u) \setminus N_{G'}(v)$,
\item else break the loop and output $H := G'[N_{G'}(u)], f_{|N_{G'}(u)}$ (where $f_{|Z}$ is the restriction of $f$ to~$Z$).
\end{enumerate}

If second item holds, it takes $O(d_{G'}(u))$ time to delete $u$ from $G'$ and update $s$ and $t$, and $O(\delta(G')-\delta(G'-\{u\})+1)$ time to update $\delta$.
Meanwhile, $\mu$ drops by $2d_{G'}(u)+1+\delta(G')-\delta(G'-\{u\})$.

If the third item holds, it takes $O(\delta-2d)$ time to delete the edges from $G'$ and update $s$ and $t$, and $O(1)$ time to set the new $\delta$.
The measure $\mu$ decreases by $2(\delta-2d)-(\delta-2d)=\delta-2d$.

Checking if the 4th item holds takes $O(d^2)$ time since after the third step, $u$ has at most~$2d$ neighbors in~$G'$, whereas in $O(d)$ time the contraction of $uv$ can be emulated by deleting $u$ and adding fewer than $d$ edges incident to~$v$.
Note that the minimum degree cannot decrease by more than one, and the new graph has one less vertex and at least one less edge.
Therefore, $\mu$ drops by at~least $2+1-1=2$.
Again, we update $s, t, f, \delta$ in $O(d)$ time.

The most time invested per one unit drop of $\mu$ is $O(d^2)$, in the third item; hence the total running time $O(d^3 n)$.
We should now check that the algorithm indeed reaches the fourth item, and that the output graph $H$ has the required properties.

Note that initially (in $G$) the average degree is $2d$, and each operation of the first three items cannot produce a~graph with average degree less than $d$ from a~graph with average degree at~least~$d$.
In particular, as $|E(G')|+|V(G')|$ strictly decreases and $G'$ cannot have less than $d+1$ vertices (as otherwise its average degree would be smaller than $d$), the algorithm eventually returns $H := G'[N_{G'}(u)]$ with $d \leq d_{G'}(u) \leq 2d$ and $|N_{G'}(u) \cap N_{G'}(v)| \geq d$ for every $v \in N_{G'}(u)$.
Therefore, $H$ has at most $2d$ vertices and minimum degree at~least~$d$.
\end{proof}

We now greedily build, in the minor given by the previous lemma, a~$(\leq 2)$-subdivision of an appropriately large clique.
If this fails, we obtain a~subgraph on $\Theta(d)$ vertices in which every pair of vertices shares $\Theta(d)$ neighbors.  

\begin{lemma}\label{lem:2-subd-or-denser}
There is an algorithm that given any graph $H$ with at~most~$2d$ vertices and minimum degree at~least~$d$, outputs:
\begin{itemize}
    \item a~$K_{\lfloor \sqrt{d}/10 \rfloor}$-minor model of~$H$ or
    \item   a~subgraph $H'$ of $H$ with at~most~$1.02 d$ vertices and minimum degree at~least~$0.94d$,
\end{itemize}
and runs in $O(d^3)$ time.
\end{lemma}

\begin{proof}
Let $X \subseteq V(H)$ be an arbitrary subset of size $s : = \lfloor \sqrt{d}/10 \rfloor$.
We denote by $\tilde{H}$ the current graph, initialized by $\tilde{H} := H$.
We in fact try to build $K_s$ as a~topological minor whose branching vertices are the vertices of~$X$, and every edge is subdivided at~most~twice.
For each pair $x \neq y \in X$, if there is an $x$--$y$ path $P$ on at~most~3 edges in $\tilde{H}-(X \setminus \{x,y\})$, add the 0, 1, or 2 internal vertices of~$P$ to the topological minor and remove them from $\tilde{H}$.
Else, output $H' := \tilde{H}[(N_{\tilde{H}}(x) \setminus X) \cup \{x\}]$ and terminate.

First, note that if none of the $s \choose 2$ iterations terminates by outputting a~subgraph~$H'$, then we have eventually built a~$K_s$-minor model.
Checking if there is such a~path of length at~most~3 between $x$ and $y$, and if so, removing its internal vertices, can be done in $O(d^2)$ time.
We repeat this at~most~${s \choose 2}=O(d)$ times, hence the total running time $O(d^3)$.

We now show that if the algorithm outputs a~subgraph $H'$, it has the required properties.
Overall, as we remove at~most~$2{s \choose 2}$ vertices from $H$, the minimum degree in every built subgraph $\tilde{H}$ is at~least~$d-2{s \choose 2} \geq 0.99 d$.
Let $x \neq y \in X$ be the (first) pair such that the algorithm outputs a~subgraph~$H'$. 
Let $A := (N_{\tilde{H}}(x) \setminus X) \cup \{x\}$ and $B := (N_{\tilde{H}}(y) \setminus X) \cup \{y\}$.
In particular, $|A| > 0.99d - s > 0.98d$ and $|B| > 0.99d - s > 0.98d$.
As $x$ and $y$ are by assumption at distance greater than 3, $A$ and $B$ are disjoint and no vertex of $A$ has a~neighbor in $B$.
Therefore, every vertex of $A$ has at~least $0.98 d -0.04d=0.94d$ neighbors in $A$.
Indeed, $|V(\tilde{H})| \leq 2d$ so there are at~most~$2d-2 \cdot 0.98d=0.04d$ vertices in $V(\tilde{H}) \setminus (A \cup B)$. 
Finally, as $A \cap B = \emptyset$ and $|B| > 0.98d$, it holds that $|A| = |V(H')| \leq 1.02d$.
\end{proof}

In the very dense graphs given by the previous lemma, it is easy to find the 1-subdivision of a~large clique as a~subgraph.

\begin{lemma}\label{lem:cm-super-dense}
Given a~graph $H'$ with at~most~$1.02d$ vertices and minimum degree at~least~$0.94d$, a~$K_{\lfloor \sqrt{d}/10 \rfloor}$-minor model of~$H'$ can be found in time $O(d^2)$.
\end{lemma}

\begin{proof}
Take an arbitrary set $X \subseteq V(H')$ of size $s := \lfloor \sqrt{d}/10 \rfloor$.
We find the 1-subdivision of $K_s$ as a~subgraph of~$H'$ with branching vertices in $X$.
For each pair $x \neq y \in X$, pick a~common neighbor of $x$ and $y$ outside $X$, add it to the topological minor, and remove it from the graph.

This can be done in time $O(d^2)$.
We just need to argue that every pair has a~common neighbor outside $X$ in the remaining graph.
In total, we remove $s \choose 2$ vertices.
Thus, at any point, the minimum degree is at~least~$0.94d - {s \choose 2} > 0.93d$.
Therefore, $x$ and $y$ have at~least $0.93d - s > 0.92d$ neighbors outside $X$.
As the total number of vertices is at~most~$1.02d$, this implies the existence of a~shared neighbor outside~$X$.
\end{proof}

We can now prove the main lemma of the section.

\begin{namedproof}[Proof of~\Cref{lem:minors-in-dense}]
We apply \Cref{lem:densifier} to the spanning subgraph of the input graph $G$ made by its first $100 h^2 n$ edges, and $d := 100 h^2$.
This yields an $H$-minor model of $G$ such that $H$ satisfies the preconditions of~\Cref{lem:2-subd-or-denser}.
The latter lemma either returns a~$K_h$-minor model (and we are done), or a~subgraph $H'$ of $H$ satisfying the preconditions of~\Cref{lem:cm-super-dense}.
In turn, this lemma yields a~$K_h$-minor model.
Unpacking the $K_h$-minor model in~$G$ is straightforward via the contraction array $f_{|V(H)}$.
The overall running time is $O(d^3 n + d^3 + d^2)=O(h^6 n)$.
\hfill$\square$\null
\end{namedproof}

\section{Conclusion}

In this paper, we gave the first algorithm that finds a balanced separator of size $O(\sqrt{n})$ in $O(n)$ time in any graph class excluding a~minor. If the smallest excluded minors have $h$ vertices (equivalently, $K_h$ is excluded), the multiplicative dependence on $h$ in both the separator size and the running time is polynomial, and more precisely, $h^{13}$. There are a few places where we can further optimize the exponent of this polynomial. For example, in \Cref{lm:connector-to-minor}, we could apply the Lovász local lemma instead of the union bound to reduce the value of $k$ in \Cref{alg:full}.
If randomness is allowed, we could use the probabilistic construction of \cite{AlonKS23}, which has a smaller dependence on~$h$. However, these optimizations would complicate our construction.
These optimizations, if implemented, are unlikely to give us a running time, say, $O(h \polylog(h)\,n)$.

For bounded-degree graphs, we can trade running time for separator size: we can construct a separator of size $O(h^{26}\sqrt{n})$ in $O(n)$ time. The idea is rather simple: we contract connected subgraphs of size $\Theta(h^{13})$ to get a new graph $H$ of size $O(n/h^{13})$, apply our separator algorithm in \Cref{thm:main} to get a balanced separator $S_H$, and then uncontract vertices in $S_H$ to get a balanced separator of $G$. The details are given in~\Cref{sec:app-linear}.

Ideally, we would want an algorithm with running time $O(h \polylog(h)\,n)$ producing a separator of size $O(h\sqrt{n})$. Such an algorithm seems to require substantially more (complicated) ideas. We leave it as an open problem for future work.

\paragraph{Acknowledgments.} We thank Jack Spalding-Jamieson for clarifying the running time of reweighted spectral partitioning, and Thatchaphol Saranurak for answering a question about the vertex expander.  This research has been initiated/conducted at the AlgUW workshop (Bedlewo 09.2025), supported by the Excellence Initiative -- Research University (IDUB) funds of the University of Warsaw. H.L.\ is supported by NSF grant CCF-2517033 and NSF CAREER Award CCF-2237288.
É.B.\ has been supported by the French National Research Agency through the project TWIN-WIDTH with reference number ANR-21-CE48-0014.
T.K.\ is supported by the European Union under Marie Skłodowska-Curie Actions (MSCA), project no.\ 101206430 \mbox{\includegraphics[height=2em]{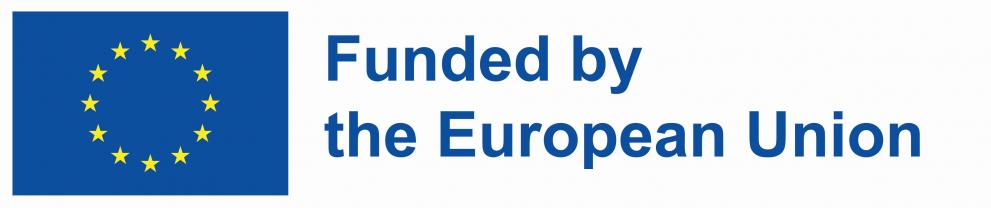}}, and by the VILLUM Foundation, Grant Number 54451, Basic Algorithms Research Copenhagen (BARC).
T.M.\ was supported by the Polish National Science Centre SONATA-17 grant number 2021/43/D/ST6/03312.
\bibliographystyle{alphaurl}
\bibliography{ref}

\appendix
\section{Proof of \Cref{lm:KPR}} \label{app:KPR-proof}

\newcommand{\Diam}{12h^2\Delta}
\newcommand{\claimqed}{\hfill$\diamondsuit$}

In this section, we will work with $\angbracket{G,w}$ and we set $W\coloneqq \sum_{v\in V(G)} w(v)$.
We restate \Cref{lm:KPR} here for convenience.

\VertexWeightKPR*

We first slightly modify the procedure $\bfs(\angbracket{G,w},v,r)$ defined in \Cref{subsec:idea}.
For simplicity, we define \EMPH{$w$-BFS} as a shortcut for $\bfs(\angbracket{G,w},v,r)$ with $r=W$ and $v$ being an arbitrary vertex unless explicitly specified otherwise.
For a $w$-BFS $r$-rooted tree $T$ of a connected graph $F$ and for every integer $j\ge1$, we define the $j$-th \emph{level} of $T$ as
\[
  L_j(T)\coloneqq
  \left\{v\in V(F):
  d_{\langle F,w\rangle}(r,v)-w(v)<j\le d_{\langle F,w\rangle}(r,v)
  \right\}.
\]
Thus every vertex $v$ occurs in exactly $w(v)$ consecutive levels.
We say that $T$ has at most $x$ levels if $L_j(T)=\emptyset$ for every $j>x$.
Observe that \Cref{lm:node-BFS} still applies in this setting.
Moreover, if a graph consists of multiple components, we can run \Cref{lm:node-BFS} component-wise, which yields the very same bound on the whole graph.
Similarly, it is easy to observe that we can run $w$-BFS from a set of vertices, still with the same running-time guarantee.

Next, we define the core objects of the KPR approach.
Given integers $\Delta,D,\ell>0$, we define a~\EMPH{$(\Delta,D,\ell)$-KPR-separator} $\Ss$ in a vertex-weighted graph $\langle G,w\rangle$ as a collection $\Ss=(S_1,\ldots,S_\ell)$ of $\ell$ vertex sets (some of which may be empty) such that:
\begin{itemize}
  \item for all $i\in[\ell]$, $|S_i|\le \frac{W}{\Delta}$, and
  \item the weak diameter of a largest connected component $C^*$ of $G-\bigcup_{i\in [\ell]}S_i$ is at most $D$.
\end{itemize}
It is readily seen that a $(\Delta,12h^2\Delta,h)$-KPR-separator matches the first outcome of \Cref{lm:KPR} with $S=\bigcup\Ss$.
 
We define a complementary object that will serve as a certificate for the second outcome.
Given integers $\Delta,D,\ell>0$, we define a \EMPH{$(\Delta,D,\ell)$-KPR-decomposition} in a vertex-weighted graph $\langle G,w\rangle$ as a collection $\Ss=(S_1,\ldots,S_\ell)$ of $\ell$ separators (vertex sets), $\ell$ rooted trees $\Tt\coloneqq(T_1,\ldots,T_\ell)$, and two vertices $a,b\in V(G)$ satisfying the following properties \hyperref[def:kpr-p1]{(P1)}--\hyperref[def:kpr-p4]{(P4)}.

Before stating the properties, we define a couple of auxiliary notions.
For every $i\in[\ell+1]$, let
\[
  H_i\coloneqq G-\bigcup_{q=1}^{i-1}S_q.
\]
Let $C^*$ be a largest connected component of $H_{\ell+1}$.
For every $i\in[\ell+1]$, let $C_i$ be the connected component of $H_i$ containing $C^*$.
In particular, $C_{\ell+1}=C^*$ and $C^*\subseteq C_{j+1}\subseteq C_j$ for $j\in[\ell]$.
For vertices $x,y\in V(T_i)$ lying on a common leaf-to-root path, let $T_i[x,y]$ denote the unique $x$--$y$ path in $T_i$.
A set $X\subseteq V(T_i)$ is \emph{convex} in $T_i$ if $T_i[x,y]\subseteq X$ whenever $x,y\in X$ lie on a common leaf-to-root path.

\begin{itemize}
  \item[(P1)] \phantomsection\label{def:kpr-p1} The component $C^*$ has weak diameter greater than $D/2$, which is witnessed by $a,b\in V(C^*)$ such that $d_{\langle G,w\rangle}(a,b)>D/2$.
  \item[(P2)] \phantomsection\label{def:kpr-p2} for all $v\in V(G)\setminus S_1$, $w(v)\le\Delta$,
  \item[(P3)] \phantomsection\label{def:kpr-p3} for all $i\in[\ell]$, $T_i$ is the $w$-BFS tree of $C_i$, rooted at a vertex $r_i$, and every connected component of $C_i-S_i$ is convex in $T_i$,
  \item[(P4)] \phantomsection\label{def:kpr-p4} for all $i\in[\ell]$, every connected component of $C_i-S_i$ is fully contained in at most $\Delta$ consecutive levels of $T_i$.
\end{itemize}

For vertices $x,y\in V(G)$, we write $d_G(x,y)\coloneqq d_{\langle G,w\rangle}(x,y)$.
For a vertex set $X\subseteq V(G)$ and $v\in V(G)$, we write
$d_G(X,v)\coloneqq\min_{x\in X}d_G(x,v)$.
We also need to measure distances in $H_i$, for which we introduce the following shortcuts.
For every $i\in[\ell]$ and $x,y\in V(C_i)$, we write $d_i(x,y)\coloneqq d_{\langle H_i,w\rangle}(x,y)$, and for the distance to the root we use $d_i(x)\coloneqq d_i(r_i,x)$.
Note that $d_1=d_G$.
We repeatedly use the following simple variant of the triangle inequality:
\begin{align}
  d_i(x,z)
  &\le d_i(x,y)+d_i(y,z)-w(y).
  \label{eq:triangle}
\end{align}

The following is a weighted variant of the moat property of Klein, Plotkin, and Rao~\cite[Lemma~4.5]{KPR93}.
It is a straightforward consequence of \hyperref[def:kpr-p3]{(P3)} and \hyperref[def:kpr-p4]{(P4)}.

\begin{lemma}[Moat property]\label{moat}
For every $i\in[2,\ell]$ and every $x,y\in V(C_i)$ such that $y$ is an ancestor of $x$ in $T_{i-1}$, $d_i(x,y)\le\Delta$.
\end{lemma}

\begin{proof}
By \hyperref[def:kpr-p3]{(P3)}, $T_{i-1}[x,y]\subseteq C_i$.
By \hyperref[def:kpr-p4]{(P4)}, this path has length at most $\Delta$.
\end{proof}

The following lemma states that, given a KPR-decomposition with sufficiently large $D$, we can extract a $K_{h,h-1}$ minor model.
In the proof of~\Cref{lm:KPR}, we will turn this model into the required $K_h$ minor model.

\begin{lemma}[Extracting $K_{h,h-1}$ minor model]\label{lm:minor}
Let $h\ge 3$ and $\Delta>0$ be integers.
If a graph admits a $(\Delta,\Diam,h)$-KPR-decomposition $(\Ss,\Tt,a,b)$, then it contains a~$K_{h,h-1}$ minor.
Moreover, a $K_{h,h-1}$-minor model can be constructed in time $O(h(m+W))$.
\end{lemma}

\begin{proof}
We set $\delta\coloneqq 6(h+1)\Delta$.

\begin{claim}\label{cl:far-terminals}
There are vertices $a_1,\ldots,a_h\in V(C^*)$ whose pairwise weak distances are at least $\delta$.
\end{claim}

\begin{proof}[Proof of~\Cref{cl:far-terminals}]
Recall that $a,b$ are two vertices in \hyperref[def:kpr-p1]{(P1)} satisfying
\[
  d_G(a,b)>\frac{\Diam}{2}=6h^2\Delta.
\]
We set $a_1\coloneqq a$ and $a_h\coloneqq b$.
Let $P$ be an $a_1$--$a_h$ path fully contained in $C^*$.
The plan is to select $a_2,\ldots,a_{h-1}$ successively along $P$.
Suppose that $a_1,\ldots,a_{j-1}$ have already been selected.
We denote $X_{j-1}\coloneqq\{a_1,\ldots,a_{j-1}\}$.
We maintain $d_G(X_{j-1},a_h)>\delta$, which is true for $j=2$ as $h\ge 2$.

Let $a_j$ be the first vertex of $P$ such that $d_G(X_{j-1},a_j)\ge \delta$.
Such a vertex exists as $a_h$ could potentially be such a vertex.
It remains to show that $d_G(a_j,a_h)>\delta$ for every $j\in[2,h-1]$.
Let $u_j$ be the predecessor of $a_j$ on $P$.
By the choice of $a_j$, there is an index $p<j$ such that $d_G(a_{p},u_j)<\delta$.
Since $u_j$ and $a_j$ are consecutive on $P$ and $w(a_j)\le\Delta$ by \hyperref[def:kpr-p2]{(P2)}, $d_G(a_{p},a_j)<\delta+w(a_j)$.
Hence, using induction and \Cref{eq:triangle}, $a_j$ satisfies the following inequality:
\begin{align*}
  d_G(a_1,a_j)
  &\le
  d_G(a_1,a_{p})
  +d_G(a_{p},a_j)
  -w(a_{p})
  <
  (p-1)\delta+w(a_{p})
  +\delta+w(a_j)-w(a_{p})\nonumber\\
  &\le (j-1)\delta+w(a_j).
\end{align*}
Consequently,
\begin{align*}
  d_G(a_j,a_h)
  &\ge
  d_G(a_1,a_h)
  -d_G(a_1,a_j)
  +w(a_j)
  >6h^2\Delta-(j-1)\delta
  \ge 6h^2\Delta-(h-2)\delta\\
  &=6(h+2)\Delta
  >\delta.
\end{align*}
By construction, $a_1,\ldots,a_h$ have pairwise weak distances at least $\delta$.
Algorithmically, the selection uses at most $h-2$ multi-source $w$-BFS runs from the vertices already selected.
\let\qed\claimqed
\end{proof}

We construct the minor model by growing the $A$-branch sets along paths to the roots of trees in $\Tt\setminus\{T_1\}$.
All tree paths in this construction lie in some $C_i$ with $i\ge2$, and hence all their vertices have weight at most $\Delta$ by \hyperref[def:kpr-p2]{(P2)}.

Let $i\ge2$ and $x\in V(C_i)$.
Let $t$ be the first vertex on the $x$--$r_i$ path in $T_i$ such that $d_i(x,t)\ge6\Delta$, and let $q$ be the first vertex such that $d_i(x,q)\ge3\Delta$ (if they exist).
We call $Q\coloneqq T_i[x,t]$ a \emph{tail in $T_i$}, $x$ its \emph{start}, $t$ its \emph{tip}, and $q$ its \emph{junction}.
A family of tails $Q_1,\ldots,Q_h$ is \emph{proper} for branch sets $A_1,\ldots,A_h,B_{i+1},\ldots,B_h$ if the tails are pairwise vertex-disjoint, each $Q_p$ starts in $A_p$, its tip lies outside every branch set, and $Q_p$ is disjoint from every branch set other than possibly $A_p$.
At step $i$, all vertices of $Q_p$ except its tip will be added to $A_p$, the tip will connect $A_p$ to $B_i$, and the junction will be the start of the tail used at step $i-1$.
By \hyperref[def:kpr-p2]{(P2)},
\begin{align}
  3\Delta
  &\le d_i(x,q)<3\Delta+w(q)\le4\Delta,
  \label{eq:djunction}\\
  6\Delta
  &\le d_i(x,t)<6\Delta+w(t)\le7\Delta.
  \label{eq:dtail}
\end{align}
Since $q$ and $t$ are ancestors of $x$ in the $w$-BFS tree $T_i$,
\begin{align}
  2\Delta
  &\le d_i(x)-d_i(q)<3\Delta,
  \label{eq:junctionDrop}\\
  5\Delta
  &\le d_i(x)-d_i(t)<6\Delta.
  \label{eq:tipDrop}
\end{align}
Moreover, every vertex $z\in V(Q)$ satisfies
\begin{equation}
  d_G(q,z)<5\Delta.
  \label{eq:tailRad}
\end{equation}
Indeed, if $z\in V(T_i[x,q])$, then by \Cref{eq:djunction},
\[
  d_G(q,z)\le d_i(q,z)\le d_i(x,q)<4\Delta.
\]
If $z\in V(T_i[q,t])$, then by \Cref{eq:junctionDrop} and \Cref{eq:tipDrop},
\[
  d_G(q,z)\le d_i(q,z)\le d_i(q,t)
  =d_i(q)-d_i(t)+w(t)<5\Delta.
\]

\begin{claim}\label{cl:tails-exist}
Let $i\ge2$, and let $x_1,\ldots,x_h\in V(C_i)$ be contained in at most $\Delta$ consecutive levels of $T_i$.
If their pairwise weak distances are at least $14\Delta$, then the tails in $T_i$ starting at $x_1,\ldots,x_h$ exist and are pairwise vertex-disjoint.
\end{claim}

\begin{proof}[Proof of~\Cref{cl:tails-exist}]
The values $d_i(x_1),\ldots,d_i(x_h)$ differ by less than $\Delta$.
If $d_i(x_p)<6\Delta$ for some $p$, then $d_i(x_k)<7\Delta$ for every $k$, and any two starting vertices would be connected through $r_i$ by a path of length less than $14\Delta$, a contradiction.
Thus all tails exist.
By \Cref{eq:dtail}, each tail has length less than $7\Delta$.
If two tails intersected, their initial subpaths to an intersection vertex would give a path of length less than $14\Delta$ between their starts.
Hence the tails are pairwise vertex-disjoint.
\let\qed\claimqed
\end{proof}

We process $i=h,\ldots,2$.
At the beginning of step $i$, we maintain:
\begin{enumerate}
  \item[(I1)] \phantomsection\label{inv:kpr-I1} a $K_{h,h-i}$-minor model contained in $C_{i+1}$, with branch sets $A_1,\ldots,A_h$ on one side and $B_{i+1},\ldots,B_h$ on the other side,
  \item[(I2)] \phantomsection\label{inv:kpr-I2} proper tails $Q_1^i,\ldots,Q_h^i$ in $C_i$, each of which is a subpath of a path to $r_i$ in $T_i$, and
  \item[(I3)] \phantomsection\label{inv:kpr-I3} if $q_p^i$ denotes the junction of $Q_p^i$, then
        \begin{equation}
          d_G(a_p,q_p^i)-w(q_p^i)<3(h-i+1)\Delta
          \qquad\text{for every }p\in[h].
          \label{eq:junctionGrowth}
        \end{equation}
\end{enumerate}
Whenever the invariant holds, \Cref{cl:far-terminals} together with \Cref{eq:junctionGrowth,eq:triangle} implies that, for distinct $p,k\in[h]$,
\begin{equation}
  d_G(q_p^i,q_k^i)
  >\delta-6(h-i+1)\Delta
  =6i\Delta.
  \label{eq:junctionSep}
\end{equation}

For $i=h$, we set $A_p\coloneqq\{a_p\}$ for every $p\in[h]$.
Since $C_{h+1}$ is contained in at most $\Delta$ consecutive levels of $T_h$ by \hyperref[def:kpr-p4]{(P4)} and $\delta>14\Delta$, \Cref{cl:tails-exist} gives pairwise vertex-disjoint tails $Q_1^h,\ldots,Q_h^h$ starting at $a_1,\ldots,a_h$.
These tails are proper: each starts in the corresponding singleton $A_p$, its tip differs from its start, and pairwise disjointness excludes every other $A$-branch set.
If $q_p^h$ is the junction of $Q_p^h$, then \Cref{eq:junctionDrop} gives $d_G(a_p,q_p^h)-w(q_p^h)<3\Delta$.
Hence the invariant holds for $i=h$.

We fix $i\in\{h,\ldots,2\}$ and suppose that the invariant holds at the beginning of step $i$.
For brevity, we write $Q_p\coloneqq Q_p^i$, and let $s_p,q_p,t_p$ be its start, junction, and tip, respectively.
We first construct $B_i$ and verify \hyperref[inv:kpr-I1]{(I1)}.
If $i>2$, we then construct the tails for step $i-1$ and verify \hyperref[inv:kpr-I2]{(I2)}--\hyperref[inv:kpr-I3]{(I3)}.

By \hyperref[def:kpr-p4]{(P4)}, the values $d_i(s_p)$ differ by less than $\Delta$, and every $x\in C_{i+1}$ satisfies $|d_i(x)-d_i(s_p)|<\Delta$.
Together with \Cref{eq:junctionDrop} and \Cref{eq:tipDrop}, this gives, for every $x\in C_{i+1}$ and all $p,k\in[h]$,
\begin{align}
  d_i(x)-d_i(q_p)
  &=d_i(x)-d_i(s_p)+d_i(s_p)-d_i(q_p)
    >-\Delta+2\Delta=\Delta,
  \label{eq:moat}\\
  |d_i(t_p)-d_i(t_k)|
  &\le |d_i(s_p)-d_i(s_k)|
    +\left|\bigl(d_i(s_p)-d_i(t_p)\bigr)-\bigl(d_i(s_k)-d_i(t_k)\bigr)\right|
    <2\Delta,
  \label{eq:tipSpread}\\
  d_i(q_p)-d_i(t_k)
  &=d_i(s_p)-d_i(s_k)
    +\bigl(d_i(s_k)-d_i(t_k)\bigr)
    -\bigl(d_i(s_p)-d_i(q_p)\bigr)
    >-\Delta+5\Delta-3\Delta=\Delta.
  \label{eq:junctionTip}
\end{align}

We define
\[
  B_i\coloneqq\bigcup_{p=1}^h V(T_i[t_p,r_i]).
\]
This is a connected subtree of $T_i$.
It is disjoint from $C_{i+1}$.
Indeed, if $x\in C_{i+1}$ and $y\in T_i[t_p,r_i]$, then, for any $k\in[h]$, \Cref{eq:junctionTip} gives
$d_i(x)>d_i(q_k)+\Delta>d_i(t_p)+2\Delta\ge d_i(y)$.

We next show that $B_i$ is disjoint from $Q_k-\{t_k\}$ for every $k$.
Suppose that $z\in(Q_k-\{t_k\})\cap T_i[t_p,r_i]$ for some $p$.
The vertex $z$ is a strict descendant of $t_k$ and an ancestor of $t_p$.
Hence, $p\neq k$.
By \Cref{eq:tipSpread},
\begin{align*}
  d_G(t_p,z)-w(z)
  &\le d_i(t_p,z)-w(z)
  =d_i(t_p)-d_i(z)
   <d_i(t_p)-d_i(t_k)
   <2\Delta.
\end{align*}
Furthermore, \Cref{eq:junctionDrop,eq:tipDrop} give
\begin{align*}
  d_G(q_p,t_p)-w(t_p)
  &\le d_i(q_p,t_p)-w(t_p)
  =d_i(q_p)-d_i(t_p)\\
  &=\bigl(d_i(s_p)-d_i(t_p)\bigr)-\bigl(d_i(s_p)-d_i(q_p)\bigr)
   <4\Delta,
\end{align*}
while \Cref{eq:tailRad} gives $d_G(z,q_k)<5\Delta$.
Concatenating these paths yields $d_G(q_p,q_k)<11\Delta$, contrary to \Cref{eq:junctionSep}.
Thus $B_i\cap(Q_k-\{t_k\})=\emptyset$ for every $k$.

For every $p$, we enlarge $A_p$ by all vertices of $Q_p$ except its tip $t_p$.
The tails were proper, and the preceding paragraphs show that the enlarged branch sets and $B_i$ are pairwise disjoint.
The last edge of $Q_p$ joins the enlarged $A_p$ to $B_i$, while all previous adjacencies are preserved because the $A_p$ only grow.
We have therefore obtained a $K_{h,h-i+1}$-minor model contained in $C_i$, with branch sets $A_1,\ldots,A_h$ and $B_i,\ldots,B_h$, concluding \hyperref[inv:kpr-I1]{(I1)} for the next step.

If $i=2$, this is the required $K_{h,h-1}$-minor model and we are done with the proof.
Hence, we may assume that $i>2$.
Since $C_i$ is contained in at most $\Delta$ consecutive levels of $T_{i-1}$ by \hyperref[def:kpr-p4]{(P4)} and the junctions $q_1,\ldots,q_h$ are pairwise at weak distance greater than $6i\Delta\ge18\Delta$, \Cref{cl:tails-exist} gives pairwise vertex-disjoint tails $Q_1^{i-1},\ldots,Q_h^{i-1}$ in $T_{i-1}$ starting at $q_1,\ldots,q_h$.
Let $z\in V(Q_p^{i-1})\cap V(C_i)$ for some $p\in[h]$.
If $z\in C_{i+1}$, then \Cref{eq:moat} gives $d_i(z,q_p)>d_i(z)-d_i(q_p)>\Delta$, contradicting \Cref{moat}.
If $z\in B_i$, then $d_i(z)\le d_i(t_k)$ for some $k\in[h]$, and by \Cref{eq:junctionTip}, $d_i(z,q_p)>d_i(q_p)-d_i(z)\ge d_i(q_p)-d_i(t_k)>\Delta$.
This is again a contradiction with \Cref{moat}.
Thus $Q_p^{i-1}$ is disjoint from both $C_{i+1}$ and $B_i$.
By \hyperref[inv:kpr-I1]{(I1)}, all branch sets present before step $i$ are contained in $C_{i+1}$, so the new tail avoids them as well.
If, for some $k\ne p$, a vertex $z$ belonged both to $Q_p^{i-1}$ and to the part of $Q_k$ just added to $A_k$, then \Cref{moat} and \Cref{eq:tailRad} give $d_G(q_p,q_k) \le d_i(q_p,z)+d_G(z,q_k)-w(z) <6\Delta$, contradicting \Cref{eq:junctionSep}.
Hence $Q_p^{i-1}$ meets no branch set other than possibly $A_p$.

Its tip does not belong to $C_i$.
Otherwise, the convexity in \hyperref[def:kpr-p3]{(P3)} would imply that the start-to-tip path is contained in $C_i$, and \hyperref[def:kpr-p4]{(P4)} would imply that it has length at most $\Delta$, whereas its length is at least $6\Delta$.
Since every current branch set is contained in $C_i$, the tip lies outside every branch set.
Thus the tails are proper, which establishes \hyperref[inv:kpr-I2]{(I2)} for step $i-1$.

Let $q_p^{i-1}$ be the junction of $Q_p^{i-1}$.
\Cref{eq:triangle,eq:junctionDrop} give
\begin{align*}
  d_G(a_p,q_p^{i-1})-w(q_p^{i-1})
  &\le
  \bigl(d_G(a_p,q_p)-w(q_p)\bigr)
  +\bigl(d_G(q_p,q_p^{i-1})-w(q_p^{i-1})\bigr)\\
  &<3(h-i+1)\Delta+3\Delta
   =3\bigl(h-(i-1)+1\bigr)\Delta.
\end{align*}
This establishes \hyperref[inv:kpr-I3]{(I3)} for step $i-1$.
Thus all invariants \hyperref[inv:kpr-I1]{(I1)}--\hyperref[inv:kpr-I3]{(I3)} hold at step $i-1$, which concludes the construction.

The selection of $a_2,\ldots,a_{h-1}$ uses at most $h-2$ multi-source runs of $w$-BFS.
At each step, the new tails are pairwise disjoint, and the union of the root paths defining $B_i$ is a subtree of one $w$-BFS tree.
Hence all objects of one step can be constructed in one pass over that tree, and the total running time is $O(h(m+W))$.
\end{proof}

Now, we are ready to prove the main statement of this section.

\begin{namedproof}[Proof of \Cref{lm:KPR}]
If $h<3$, the problem is trivial.

It suffices to prove the statement for connected graphs.
Indeed, we apply the connected-case procedure independently to every connected component $F$ of $G$, with the restriction of $w$ to $V(F)$.
If one invocation returns a $K_h$-minor model, then this is also a $K_h$-minor model of $G$.
Otherwise, let $(C_F^*,S_F)$ be the pair returned for $F$.
We set $S\coloneqq\bigcup_F S_F$ and choose $C^*$ among the components $C_F^*$ with maximum cardinality.
It is straightforward to verify that this meets the requirement of the lemma.
Hence, in the remainder of the proof, we assume that $G$ is connected.

By \Cref{lm:node-BFS}, we construct a $w$-BFS tree $T_1$ rooted at an arbitrary node $v$ of $G$.
If there are at most $\frac{\Diam}{2}$ levels in $T_1$, we stop the procedure, as $G$ has weak diameter at most $\Diam$.
In that case, $S\coloneqq \emptyset$ and $C^*\coloneqq G$.
Otherwise, we show how to construct in linear time a $(\Delta,\Diam,h)$-KPR-separator~$\Ss$ or a $(\Delta,\Diam,h)$-KPR-decomposition $(\Ss,\Tt,a,b)$.
In the latter case, \Cref{lm:minor} gives a $K_{h,h-1}$-minor model with sides $A_1,\ldots,A_h$ and $B_2,\ldots,B_h$.
For every $i\in\{2,\ldots,h\}$, choose an edge joining $A_i$ to $B_i$ and contract it.
The resulting $h$ branch sets, namely $A_1$ and $A_i\cup B_i$ for $i\in\{2,\ldots,h\}$, form a $K_h$-minor model.
In the former case, we return $\Ss$, for which $C^*$ has weak diameter at most $\Diam$.

Hence, there are more than $\frac{\Diam}{2}$ levels in $T_1$.
Given a $w$-BFS tree $T$ of a connected component, we define $S\in \Ss$ as the cheapest out of the following possible separators $Z_i(T)\coloneqq \bigcup_{q\ge0}L_{i+q\Delta}(T)$ for $i\in[\Delta]$.
In the first round, let $S_1$ be a minimum-cardinality set among $Z_1(T_1),\ldots,Z_\Delta(T_1)$.
Observe that, as the total vertex weight is $W$, which corresponds to the total number of occurrences of the vertices in all levels, the cheapest of the proposed $\Delta$ separators has size at most $\lfloor \frac{W}{\Delta}\rfloor$.
Moreover, if a vertex $v\in V(G)$ has $w(v)>\Delta$, then it belongs to all the separators defined above and will not be part of $G-S_1$, which yields \hyperref[def:kpr-p2]{(P2)}.
As a result of this step, we are left with a set of connected components of $G-S_1$ denoted by $\Cc_1$.
Now, we repeat the same procedure for each component in $\Cc_1$ separately and recursively.
That means that, in step $j\le h$, we consider all unmarked connected components of the graph $H_j$.
We run $w$-BFS for each separately.
If such a connected component has at most $\frac{\Diam}{2}$ levels, we do nothing, as its weak diameter is even smaller.
That component is marked and will not be considered in any later round.
If all connected components are marked, then, setting $S_j=\cdots=S_h=\emptyset$, we return $S\coloneqq S_1 \cup \ldots \cup S_{j-1}$ together with a largest component $C^*$ as a $(\Delta,\Diam,h)$-KPR-separator.
If a connected component with $w$-BFS tree $T$ has more levels, we apply the cutting procedure described above by choosing the smallest cut $Z_i(T)$ among the options $i\in[\Delta]$.
Then we set $S_j$ as the union of the smallest-size cuts over all connected components considered in that step.
The analysis still holds, as in each component we take a cut of size at most the average over the $\Delta$ options.

After completing step $h$, we have a set of connected components $\Cc_h$.
We check which component has the largest size, and we denote it as $C^*$.
This can be easily checked by a component-wise $w$-BFS check in total time $m+W$.
If $C^*$ was marked earlier, we return $(C^*,\Ss)$.
Otherwise, we run $w$-BFS on $G$ starting at a vertex $a\in V(C^*)$.
If all vertices of $C^*$ appear within the first $\Diam/2$ levels of the resulting tree, then $C^*$ has weak diameter at most $\Diam$.
Hence, we have the right KPR-separator and so we return $(C^*,\Ss)$.
Otherwise, there are vertices $a,b\in V(C^*)$ that witness \hyperref[def:kpr-p1]{(P1)}.
For the other properties, we take the $w$-BFS trees built during the procedure for the residual components containing $C^*$ and set them as $\Tt$.
Observe that these trees exist because the connected component containing $C^*$ was cut in every round of the procedure: a component was left unchanged only if the corresponding tree had at most $\Diam/2$ levels, and such components were not considered again.
Then \hyperref[def:kpr-p3]{(P3)} follows because these trees are $w$-BFS trees.
Convexity follows because an ancestor path between two vertices of a connected component cannot cross a cut layer.
Property \hyperref[def:kpr-p4]{(P4)} follows because each residual component is contained between two cut levels that are $\Delta$ apart.
Therefore, $\Tt$, $\Ss$, and the vertices $a,b\in V(C^*)$ form a $(\Delta,\Diam,h)$-KPR-decomposition.

The running time matches the statement because we invoked the $w$-BFS algorithm component-wise at most $h$ times (\Cref{lm:node-BFS}), with each invocation running in total time $O(m+W)$.
In addition, we run one more $w$-BFS on the whole graph when determining $C^*$.
If we return the $K_h$-minor model, we run \Cref{lm:minor} to construct a $K_{h,h-1}$-minor model and then perform the $h-1$ contractions described above.
The total additional running time is $O(h\cdot(m+W))$.
\hfill$\square$\null
\end{namedproof}

\section{Proof of~\Cref{lm:subcubic}}\label{sec:subcubic}

We restate \Cref{lm:subcubic} and then prove it.

\almostembeddinglem*
\begin{proof}
Let $H$ be the graph obtained by subdividing every edge of $K_t$ twice, i.e., replacing the edges of $K_t$ by paths of $3$ edges.
We have that $|V(H)|+|E(H)| = t + 5 \cdot \binom{t}{2} \le 3 t^2$.
We claim that an almost-embedding of $H$ can be turned into a $K_t$-minor model $\{C_1, C_2, \ldots, C_t\}$ in $O(t^2 m)$ time.

Denote the almost-embedding by $\phi$, and the vertices of $H$ corresponding to the vertices of $K_t$ by $v_1, \ldots, v_t$.
First, we set $C'_i = \bigcup_{u \in N(v_i)} V(\phi(v_i u))$, i.e., it is the union of the paths of $\phi$ corresponding to the edges incident to~$v_i$.
The properties of $\phi$ guarantee that (1) each $G[C'_i]$ is connected, and (2) the sets $C'_i$ are pairwise disjoint.

Then, consider a pair $i,j$ with $1 \le i < j \le t$.
Let $v_i, a_{i,j}, b_{i,j}, v_j$ be the path of three edges in $H$ between $v_i$ and~$v_j$.
We have that $\phi(a_{i,j} b_{i,j})$ is a path that intersects both $C'_i$ and $C'_j$.
We let $C'_{i,j}$ be the internal vertices of a~minimal subpath of $\phi(a_{i,j} b_{i,j})$ that intersects both $C'_i$ and $C'_j$.
In particular, $C'_{i,j}$ can be empty if $C'_i$ and $C'_j$ are adjacent.
It holds that (1) each $G[C'_{i,j}]$ is connected or empty, (2) each $C'_{i,j}$ either is adjacent to both $C'_i$ and $C'_j$, or is empty, in which case $C'_i$ and $C'_j$ are adjacent, and (3) the sets $C'_{i,j}$ are pairwise disjoint and disjoint from the sets~$C'_k$ for all $k$.

Now, taking $C_i = C'_i \cup \bigcup_{j > i} C'_{i,j}$ results in $\{C_1, C_2, \ldots, C_t\}$ being a $K_t$-minor model.
This construction can clearly be implemented in $O(t^2 m)$ time.
\end{proof}

\section{Balanced Separators of Size $O(\poly(h) \sqrt n)$ in Time $O(n)$ in Bounded-Degree $K_t$-Minor-Free Graphs}\label{sec:app-linear}

\begin{lemma}\label{lem:connected-partition}
There is an algorithm that, given any connected graph $G$ of maximum degree $\Delta$ and positive integer $p \leqslant |V(G)|$, builds a~partition $\mathcal P = \{P_1, \ldots, P_s\}$ of $V(G)$ such that for every $i \in [s]$
\begin{itemize}
\item $G[P_i]$ is connected, and
\item $p \leq |P_i| \leq p (\Delta+1)$,
\end{itemize}
and runs in time $O(n \Delta)$.
\end{lemma}

\begin{proof}
In $O(n \Delta)$ time run a~depth-first search (DFS) on $G$~\cite{Tarjan72}, from an arbitrary vertex~$u$.
Let $T$ be the resulting rooted spanning tree.
In $O(n)$ time, mark each node $v$ of $T$ whenever the subtree of $T$ rooted at $v$ has fewer than~$p$ nodes.
This can be done in a bottom-up way by computing the size of each subtree (and capping at~$p$).
Let $Q$ be the leftmost root-to-leaf branch of~$T$.

We build the first part of $\mathcal P$ as follows.
Initialize $x := u$ and $X := \{u\} \cup Y$, where $Y$ is the union of the vertex sets of the subtrees of~$T$ rooted at a~marked child of~$x$.
While $|X| < p$, set the new $x$ as the first node appearing after (the old) $x$ in the preorder traversal of~$T$ that is not yet in~$X$, add it to~$X$ as well as every subtree of~$T$ rooted at a~marked child of~$x$.
When exiting the while loop, the size of $X$ is between $p$ and $p \Delta$.
We set $P_1 := X$.
It is indeed connected, and $p \leq |P_1| \leq p (\Delta+1)$.

We then proceed inductively in each connected component of~$T-P_1$ (observe that several such components may well be a~single component in~$G-P_1$) and same marking.
Note that each such component has at~least~$p$ vertices, so the while loop of each call terminates.
Computing a~new part $P$ of $\mathcal P$ takes $O(|P|)$ time.
The partition $\mathcal P$ is thus computed in $O(n \Delta)$ total time.
\end{proof}

\begin{theorem}\label{thm:purely-linear} 
Let $G$ be a~given graph with $n$ vertices and maximum degree $\Delta=O(1)$, and $h > 0$ be a parameter. There is an algorithm that runs in deterministic $O(n)$ time and outputs either:
\begin{itemize}
    \item a balanced separator of size $O(h^{26} \sqrt{n})$ or
    \item  $\bot$ if $G$ contains a $K_h$ as a minor. 
\end{itemize}
 Furthermore, in the latter case, we can output a $K_h$-minor model of $G$ with probability at least $1/2$ in an additional randomized $O(n)$ time.
\end{theorem}

\begin{proof}
Apply \Cref{lem:connected-partition} to $G, p := h^{13}, \Delta$, and get, in $O(n)$ time, a~partition $\mathcal P$ satisfying the lemma. 
In further $O(n)$ time, compute the graph $G'$ obtained by contracting each part of $\mathcal P$ into a~single vertex.
As each part of $\mathcal P$ is connected, $G'$ is a~minor of~$G$.
Note that $G'$ has at~most $n/h^{13}$ vertices.

Now run \Cref{thm:main} on $G'$.
This takes time $O(h^{13} n/h^{13})=O(n)$.
If this yields a~$K_t$-minor model, the partition~$\mathcal P$ lifts it to a~$K_t$-minor model in~$G$.
If, instead, we get a~balanced separator of size $O(h^{13} \sqrt n)$ in $G'$, this gives a~$(1-\frac{1}{2\Delta+3})$-balanced separator of size $O(h^{26} \sqrt n)$ in $G$. 
With a~constant number of iterations (recall $\Delta=O(1)$), we get a~$K_h$-minor model in~$G$ or a~balanced separator of size $O(h^{26} \sqrt n)$.
\end{proof}

\end{document}